\documentclass[11pt, letterpaper]{article}
\usepackage[utf8]{inputenc}
\usepackage{fullpage,times}
\usepackage[T1]{fontenc}
\usepackage{authblk}
\usepackage[margin=1in]{geometry}
\usepackage{amsmath,amssymb,amsthm}
\usepackage{hyperref}
\usepackage{graphicx} 
\usepackage{amsmath}
\usepackage{amssymb}
\usepackage{amsthm}
\usepackage{thmtools}
\usepackage{thm-restate}
\usepackage{comment}
\usepackage{mathtools}
\usepackage{csquotes}

\usepackage{graphicx}
\usepackage{enumerate}
\usepackage{float}
\usepackage{xcolor}
\usepackage{tikz}
\usepackage{algpseudocode}
\usepackage{algorithm2e}
\usepackage{soul}
\RestyleAlgo{ruled}
\SetKwComment{Comment}{/* }{ */}
\renewcommand{\bar}{\overline}
\usepackage{setspace}
\usepackage{cleveref}
\declaretheorem[name=Theorem,numberwithin=section]{thm}
\newtheorem{theorem}{Theorem}[section]

\newtheorem{question}{Question}
\newtheorem{lemma}[theorem]{Lemma}
\newtheorem{claim}[theorem]{Claim}
\newtheorem{proposition}[theorem]{Proposition}
\newtheorem{observation}[theorem]{Observation}
\newtheorem{corollary}[theorem]{Corollary}

\newtheorem{definition}{Definition}
\newtheorem{fact}[theorem]{Fact}
\newtheorem*{remark}{Remark}

\newenvironment{proofof}[1]{\emph{Proof of #1.  }}{\hfill$\Box$}

\newcommand{\EGnote}[1]{{\textcolor{green}{Elazar: #1}}}

\newcommand{\abs}[1]{\lvert #1 \rvert}

\newcommand{\N}{\mathbb{N}}
\newcommand{\A}{\mathcal{A}}
\newcommand{\ALG}{\mathrm{ALG}}
\newcommand{\OO}{\mathcal{O}}

\newcommand{\E}{\mathcal{E}}
\newcommand{\I}{\mathcal{I}}
\newcommand{\ED}{\Delta_{edit}}
\newcommand{\TED}{\Tilde\Delta_{edit}}
\newcommand{\LCS}{\Delta_{indel}}
\newcommand{\TLCS}{\Tilde\Delta_{indel}}

\usepackage{amsmath,amssymb,amsthm}
\usepackage{hyperref}
\usepackage{graphicx} 
\usepackage{amsmath}
\usepackage{amssymb}
\usepackage{amsthm}
\usepackage{thmtools}
\usepackage{thm-restate}
\usepackage{comment}
\usepackage{mathtools}
\usepackage{csquotes}
\usepackage{graphicx}
\usepackage{enumerate}
\usepackage{float}
\usepackage{xcolor}
\usepackage{tikz}
\usepackage{algpseudocode}
\usepackage{algorithm2e}
\usepackage{soul}
\RestyleAlgo{ruled}
\SetKwComment{Comment}{/* }{ */}
\renewcommand{\bar}{\overline}
\usepackage{setspace}
\usepackage{cleveref}

\newcommand{\emptystr}{\Lambda}

\title{Many Flavors of Edit Distance} 

\author[1]{Sudatta Bhattacharya\thanks{Email: sudatta@iuuk.mff.cuni.cz. Partially supported by the project of Czech Science Foundation no.
19-27871X, 24-10306S and by the project GAUK125424 of the Charles University Grant Agency.}}
\author[2]{Sanjana Dey\thanks{info4.sanjana@gmail.com}}
\author[3]{Elazar Goldenberg\thanks{elazargo@mta.ac.il}}
\author[1]{Michal Kouck{\'{y}}\thanks{Email: koucky@iuuk.mff.cuni.cz. Partially supported by the Grant Agency of the Czech Republic under the grant agreement no. 19-27871X.}}
\affil[1]{Computer Science Institute of Charles University,
Malostransk{\'e}  n{\'a}m\v{e}st\'{\i} 25,
118 00 Praha 1, Czech Republic}
\affil[2]{University of Singapore, Singapore}
\affil[3]{The Academic College of Tel-Aviv-Yaffo, Israel}

\begin{document}
\maketitle
\begin{abstract}
    
    Several measures exist for string similarity, including notable ones like the edit distance and the indel distance. The former measures the count of insertions, deletions, and substitutions required to transform one string into another, while the latter specifically quantifies the number of insertions and deletions. Many algorithmic solutions explicitly address one of these measures, and frequently techniques applicable to one can also be adapted to work with the other. In this paper, we investigate whether there exists a standardized approach for applying results from one setting to another. Specifically, we demonstrate the capability to reduce questions regarding string similarity over arbitrary alphabets to equivalent questions over a binary alphabet. Furthermore, we illustrate how to transform questions concerning indel distance into equivalent questions based on edit distance. This complements an earlier result of Tiskin (2007) which addresses the inverse direction.
\end{abstract}

\section{Introduction}
String-related metrics, such as edit distance, longest common subsequence distance, are pivotal in numerous applications that deal with text or sequence data.
Edit distance metric, also referred to as Levenshtein distance~\cite{Lev66}, quantifies the minimum number of single-character edits (insertions, deletions, or substitutions) necessary to convert one string into another. This metric finds extensive application in spell checking and correction systems, DNA sequence alignment and bioinformatics for comparing genetic sequences, as well as in natural language processing tasks such as machine translation and text summarization, among other fields.
The longest common subsequence metric, also known as {\em Indel} or {\em LCS distance}, evaluates the disparity between two strings by determining the minimum number of single-character edits while forbidding substitutions. This metric has diverse applications, including text comparison and plagiarism detection, DNA and protein sequence analysis for recognizing shared regions or motifs, music analysis to uncover similarities between musical sequences, and document clustering and classification based on content similarity. 

From a computational complexity perspective computing distances under the Edit or Indel metric typically takes quadratic time complexity, as initially demonstrated by Wagner~\cite{wagner1974}. Subsequent research has marginally improved this complexity by reducing logarithmic factors, as evidenced by Masek and Paterson~\cite{Masek80} and Grabowski~\cite{Grabowski16}. Additionally, Backurs and Indyk~\cite{BI15} demonstrated that a truly sub-quadratic algorithm  \(O(n^{2-\delta})\) for some \(\delta>0\) would lead to a \(2^{(1-\gamma)n} \)-time algorithm for CNF-satisfiability, contradicting the Strong Exponential Time Hypothesis. Similarly, Abboud et al.~\cite{ABW15} established a similar result for computing the Indel metric between string pairs.  Notably, obtaining an efficient isometric embedding for the edit metric into the Indel metric would effortlessly yield the latter result.

Extensive research has been conducted on approximating edit distance, with studies dating back to the work of Landau et al.~\cite{LMS98,BYJKK04,BES06,AO09,BEKMRRS03,BEGHS21,CDGKS20,GRS20,KS20,BR20}, ultimately culminating in the breakthrough result of Andoni and Nosatzki~\cite{AN20}, which offers a (large) constant factor approximation in nearly linear time. However, approximating the Indel distance has not received similar attention, and although one expects the same techniques should provide similar results for Indel distance, one would need to check all the details of the construction to verify the exact properties of such a result.
This exhibits a general pattern where results for one of the measures can often be adapted for the other but there is no simple tool that would guarantee such an automatic transformation. 

Similar pattern emerges when dealing with the string measures over different size alphabets. For example, the hardness result elucidated by Backurs and Indyk~\cite{BI15} is constrained to some ``large'' constant-size alphabets. Subsequent research has revealed that the computational task of computing edit distance is also hard for binary alphabets using an ad hoc approach. Once more, the possibility of achieving an efficient isometric embedding between strings residing in large alphabets and those in smaller ones could potentially resolve these questions automatically. Another scenario where the alphabet size becomes relevant is in the simple linear-time approximation algorithm for the length of the LCS. The naive algorithm provides a \(\abs{\Sigma}\)-approximation, where \(\Sigma\) represents the alphabet to which the strings reside. Therefore, if one can embed strings from a large alphabet into those from a smaller alphabet while approximately preserving distances, it may lead to an improvement in the approximation factor.
Given the current circumstances, we pose the following questions, which we then delve into extensively:
\begin{question}
    Does an isometric embedding exist between the edit metric and the Indel metric? Can it be computed efficiently?
\end{question}
\begin{question}
    Does an isometric embedding exist between edit metric on arbitrary alphabets and the edit distance on binary alphabets? Can it be computed efficiently?
\end{question}

\subsection{Our Contribution}
This paper introduces multiple mappings that establish connections between various string metrics.  Specifically, we transform points residing within a designated input metric space \((M,d) \) into points within an output metric space \((M',d')\), ensuring that the distance between any pair of points in \(M\) under \(d\), is (approximately) preserved on the output pairs under \(d'\). In this paper, we introduce a relaxation of the concept of isometric embedding, permitting scaling factors, as outlined below.

We say an embedding  $E:M\to M'$ is \textbf{scaled isometric} if there exists a function $f:\N\to \N$ that maps distances from the original space to the embedded one, such that for every pair of points \(x,y\in M\) we have: 
\(d'(E(x),E(y))=f(d(x,y))\). This implies that such an embedding can be utilized to reduce the computation of distances in the original metric $M$  to computing distances under $M'$, provided that $f$ is invertible and computationally tractable.

A special case of scaling involves preserving the normalized distances. That is, for any metric on a string space, it is intuitive to define the normalized distance between a pair of strings of a specified length, as their distance divided by the maximum distance of pairs of that same length. Subsequently, for an embedding \(E:M\to M'\) that preserves lengths (i.e., the length of the output string is a function of the input string), we further assert that it is \textbf{normalized scaled isometric} if, for every pair from \(M\), the normalized distance of the embedded strings preserves the normalized distance of the original ones.  

Utilizing normalized scaled isometric embedding can facilitate the computational process of approximating distances in the original metric \(M\)  to compute approximate distances under \(M'\), even if the normalized distances are not perfectly preserved by the embedding but distorted by a (multiplicative) factor of \(c\). In such a case, any \(c'\) approximation algorithm for the metric \(M'\) could thus be transformed into a \(c\cdot c'\) approximation algorithm for \(M\).

\subsubsection{Our Results}
In the sequel, we utilize the notation \(\LCS\) to represent the Indel distance between a pair of strings, \(\ED\) to denote their edit distance. Additionally,  we use \(\TLCS\) and \(\TED\) to represent the normalized Indel distance between a pair of strings (for precise definitions, refer to Section~\ref{sec:preliminaries}).

\textbf{Alphabet Reduction - Succinct Embedding:}

Our first result pertains to alphabet reduction achieving normalized scaled isometric embedding. In this context, as elaborated  in Section~\ref{sec:alphabetReduction},  any embedding contracts the normalized distances of some of the pairs.  Consequently, we shift our focus to approximate normalized scaled isometric embedding, where we permit slight distortions in the normalized distances. Our main result is outlined below:

\begin{restatable}[Alphabet Reduction - Succinct Embedding]{thm}{AlphabetReductionNormalized}
\label{thm:AlphabetReductionNormalized}

Let $\Gamma $ be a finite alphabet, and let $0<\varepsilon<1/4$.
    There exists an alphabet $\Sigma$, where \(\abs {\Sigma} = O (\frac{1}{\varepsilon^2})\) and there exists \(E:\Gamma^*\to \Sigma^*\) satisfying: 
\[\forall X,Y \in \Gamma^*: \, \, \,\TLCS(E(X),E(Y))\in \left[(1-\varepsilon) \TLCS(X,Y), \TLCS(X,Y)\right].\]

Moreover, for every $X\in \Gamma^n$ we have: \(\abs{E(X)}=O(n\log (\abs{\Gamma}))\).
\end{restatable}

Observe that embedded string length is optimal up to logarithmic factors. The size of the alphabet $\Sigma$ must exceed $1/\varepsilon$, as otherwise, by a claim proved later, 
there will be a pair of strings whose distance will be contracted by at least  \( \left(1-\frac{1}{\abs{\Sigma}} \right)\)-factor. 

As a result, we find that when focusing on $\LCS$ approximation, we can, without loss of generality, limit our scope to a constant alphabet size without significantly impacting the quality of the approximation. This is captured in the following corollary which we provide without a proof. 

\begin{corollary}
Let $\Gamma $ be a finite alphabet, and let $0<\varepsilon<1/4$. 
Suppose that there exists an algorithm $\ALG$ that provides a $c$-factor approximation for the $\LCS$ metric for any pairs of strings in $\Sigma$ of total length $N$, where $\Sigma=O(1/\varepsilon^2)$, running in time $t(N)$.

Then there exists an algorithm $\ALG'$ that, given any pair of string in $\Gamma$ of total length $n$, provides a $c+\varepsilon$-approximation for their $\LCS$ distance in time $t(n\log \abs{\Gamma})$. 
    
\end{corollary}

The proof strategy of Theorem~\ref{thm:AlphabetReductionNormalized} is as follows: Initially, we construct an error-correcting code within the smaller alphabet  \(\Sigma\), where \(\abs{\Sigma}=O(1/\varepsilon^2)\). This code ensures that every distinct pair of codewords shares only a small portion of LCS, indicating a significant \(\LCS\) between them. The code comprises $\abs{\Gamma}$ words,  with its dimension being logarithmic in the size of $\Gamma$. We interpret this code as a mapping from characters within \(\Gamma\) to short strings in \(\Sigma\).  The existence of such a code can be demonstrated using the probabilistic method. The construction is almost tight in terms of the smaller alphabet size: it is not hard to show that for any code \(C\subseteq \Sigma^k\), if \(\abs{C}>\abs{\Sigma}\), then there exist $c\neq c' \in C$, satisfying: \(\LCS(c,c')<(1-\frac{1}{\abs{\Sigma}})2k.\)

Employing the code construction, we embed input strings in a straightforward and natural manner:  Consider a string residing in \(\Gamma^*\), then each of its characters  is sequentially encoded using the code. This encoding ensures that identical characters are mapped to the same codeword, while the code's distance guarantees  that distinct characters may have only a few shared matches, akin to the global nature of \(\LCS\) alignments.  If there were no matches between distinct characters' encodings, any alignment for the embedded strings could be converted into an alignment between the input strings without any distortion in the normalized costs. However, these few matches between non-matching codewords introduce a slight distortion and complicate the proof.

\textbf{Alphabet Reduction - Binary Alphabets:}

Our previous method relied on a local approach, wherein the characters of the string from the large alphabet were encoded sequentially and independently. It appears that such a local strategy may not result in a normalized scaled isometric embedding into a small, particularly binary alphabet. As codes in such alphabets  have a relative distance of at most \(1/2\), and this distance affects the distortion of the normalized distances. Therefore, achieving alphabet reduction into binary alphabets requires a scaling that is not normalized, as well as a more global approach that does not encode the characters sequentially and independently.  However, there's a caveat: the encoding still needs to be performed independently on each string rather than on a pair of strings. Specifically, if we take a particular string \(X\) its encoding will remain the same regardless of the second string \( Y\). This aspect forms the focal point of our forthcoming result.

For clarity, we present our results for the \(\LCS\) metric, with a similar approach applicable to the \(\ED\) metric. Our main result states that one can embed \(\LCS\)  over any alphabet \(\Sigma\) into binary alphabet. We employ asymmetric embedding and prove that the distances are preserved up to some scaling function. The dimension of the embedded strings is quasi-polynomial, making this result more of a proof of concept at the moment. Nonetheless, we find it conceptually intriguing  and pose the question of decreasing the target dimension as an open question. Our main result is as follows:

\begin{thm}[Informal statement of Theorem~\ref{thm:formulas}]\label{thm:binaryAlphabet}
For any alphabet $\Sigma$ and for every $n\in \N$ 
there exist functions $G,H : \Sigma^n \rightarrow \{0,1\}^N$, \(f:[n]\to [N]\)   where $N=n^{O(\log n)}$
such that  for any $X,Y \in \Sigma^n$,
\[ \LCS(G(X),H(Y))=f(\LCS(X,Y)).\]    
\end{thm}

 Our initial consideration revolves around the fact that deciding whether \(\LCS(X,Y) \le k\) for two strings \(X,Y\in \Sigma^n\) and a threshold parameter $k$, can be accomplished by a Turing machine utilizing \(\log(n)\)-space. This capability can then be translated into a formula of \(\log^2(n)\)-depth. 

The foundation of our construction converts this formula into a pair of binary strings \(X',Y'\) of quasi-polynomial length, where \(X'\) (respectively, \(Y'\)) depends solely on \(X\) (\(Y\)) and the Indel distance between \(X',Y'\) is contingent on the distance between  \(X,Y\).  This is achieved by recursively transforming the formula into such a pair of strings gate by gate. Two essential components, referred to as \(AND\)- and \(OR\)-gadgets, implement this process.

The input for the \(AND\)-gadget consists of two pairs of strings \((X_0,Y_0), (X_1,Y_1)\), which can be thought of as outputs from previous levels. Here the \(X_i\)'s only depend on \(X\) and the \(Y_i\)'s only depend on \(Y\). Moreover, we are guaranteed that \(\LCS(X_i,Y_i)\) can only take two values \(\{ 
F,T\}\), where \(F<T\). The goal is to concatenate the \(X_i\) into a single string \(X\) and the \(Y_i\)'s into a different string \(Y\) such that: \(\LCS(X,Y) \) can also take value in \{$F',T'$\}, where \(F' < T'\) and: 
\(\LCS(X,Y)=T'\) if and only if \(\LCS(X_0,Y_0) =T \wedge \LCS(X_1,Y_1) =T\). Similarly for the \(OR\)-gadget.
Our construction of the LCS instance from the Formula Evaluation is similar to that of Abboud and Bringmann~\cite{DBLP:conf/icalp/AbboudB18} which considers reduction of the Formula Satisfiability to LCS.
The purpose of our reduction is different, though.

\textbf{Scaled Isometric Embedding of Indel into Edit Metrics:}
While our previous results aimed on reducing the alphabet size while keeping the underlying metric (either \(\LCS\) or \(\ED\)), this section focuses on converting one metric into another. 
Tiskin~\cite{Tiskin08} (in section 6.1) proposed a straightforward embedding from the $\ED$ metric to the $\LCS$ metric. 
This inspired our exploration into embedding in the reverse direction i.e. from the $\LCS$ metric to the $\ED$ metric.  Our primary contribution in this realm is a scaled isometric embedding from the $\LCS$ metric to $\ED$,  as outlined below. 
\begin{thm}[Indel Into Edit Metrics Embedding - Approximate embedding -- Statement of Theorem \ref{thm:indel-edit-apx}]\label{thm:IndelEditApx}
     For any alphabet $\Sigma$, $n \in \N $ and $ \varepsilon \in (0,1]$, there exist mappings $E:\Sigma^n\to \Sigma^n$ and $E':\Sigma^n\to (\Sigma\cup \{\$\})^{N}$, where $N = \Theta(n/\varepsilon)$, such that for any $X,Y \in \Sigma^n$, we have
     $$\ED(E(X),E'(Y)) = N - n+k, \text{ where } k\in \left[\frac{\LCS(X,Y)}{2},(1 + \varepsilon)\frac{\LCS(X,Y)}{2}\right).$$
     

\end{thm}

Observe that while plugging \(\varepsilon\le\frac{1}{n}\) we obtain a scaled isometry at the expense of a quadratic increase in the length of the second string. Conversely, for constant values of \(\varepsilon\), $N$ scales as $O(n)$  albeit with the trade-off of only approximately preserving distances within a constant factor. For intermediate values of \(\varepsilon\), we can compromise between the accuracy and the stretch length.


Let us revisit Tiskin's (section 6.1 of \cite{Tiskin08} construction of the reverse embedding, namely, from $\ED$ into $\LCS$. The embedding proceeds as follows: a special character $\$$ is appended after every symbol of each string. It is easy to check that for each pair of strings, the \(\LCS\) between the embedded pair of strings equals twice the \(\ED\) between the original pair. 
In our construction, one string remains unaltered, while for the second string, we append after every symbol a block of length \(n\) consisting of the special character. 

The core of the proof demonstrates the conversion of any \(\LCS\)-alignment for the input strings into an \(\ED\)-alignment for the embedded strings, preserving the distances up to a scaling factor. This process involves replacing any deletions originally performed on the first string by substituting the characters with the special inserted character. Deletions made on the second string remain unaffected.

\subsection{Related Work}
The problem of embedding edit distance into other distance measures, such as Hamming distance, $\ell_1$, etc., has attracted significant attention in the literature. Let us briefly survey some of these approaches. 

Chakraborty et al.~\cite{CGK16} introduced a randomized embedding scheme from the edit distance to the Hamming distance. This embedding transforms strings from a given alphabet into strings that are three times longer. For each pair of strings embedded using the same random sequence, with high probability the edit distance between the embedded strings is at most quadratic in the Hamming distance of the original strings.
Batu et al.~\cite{BES06} introduced a dimensionality reduction technique: Given a parameter \(r>1\), they reduce the dimension by a factor of \(r\) at the expense of distorting the distances by the same factor. They employed the locally consistent parsing technique for their embedding. 
Ostrovsky and Rabani \cite{hamming_to_l1_rabani_2007} presented a polynomial time embedding from edit distance to $\ell_1$ distance with a distortion of $\OO(2^{\sqrt{\log n \log \log n}})$. They proposed a randomized embedding where the length of the output strings is quadratic in the input strings, and the distances are preserved, with high probability, up to the distortion factor.

\subsection{Future Directions}

\textbf{Introducing a Robust Concept of Approximation: Transitioning from Approximating \(\ED\) into Approximating \(\LCS\)}

Recall  that one of the reasons we aimed to isometrically embed the  \(\LCS\) metric into the \(\ED\) metric stemmed from the abundance of approximation results for \(\ED\) that might not easily extend to the \(\LCS\) metric. A natural approach, based on our embedding result, is to approximate the \(\LCS\) distance between \(X,Y\) by approximating the \(\ED\) between the embedded strings. However, this is not an immediate consequence due to the substantial disparity in length between the embedded strings and the notion of approximation in this case, as detailed next:

Recall that in Theorem~\ref{thm:IndelEditApx} the scaling mechanism is not normalized, i.e, the embedding function did not preserve normalized distances, but instead:
\[ \ED(E_1(X),E_3(Y)) = N - n+k, \text{ where } k\in \left[\frac{\LCS(X,Y)}{2},\dots, (1 + \varepsilon)\frac{\LCS(X,Y)}{2}\right]\] 
where \(N,n\) are the length of the embedded strings, and \(N=\Theta(\frac{n}{\varepsilon})\). Observe that the \(\ED\) between the embedded strings lies in the range of \([N-n,N]\).

Considering the substantial difference in length between the embedded strings, an algorithm that consistently outputs the value \(N-n\), regardless of the embedded strings, already yields a  \(1+O(\varepsilon)\)-approximation for the distance between the embedded strings. Certainly, such an outcome provides no information about the \(\LCS\) of the original strings.
Therefore, we introduce a more robust notion of approximation that generally addresses the discrepancy in string lengths:

\begin{definition}[A Robust Notion of Approximation:]
    Let \(c>1\), let \(\Sigma\) be a finite set, and let \(X,Y\in \Sigma^*\).  Define \(\abs{X}=N, \abs{Y}=n\), and assume \(N\ge n\). Define $k_{X,Y}$ such that: \(\ED(X,Y)=N-n+k_{X,Y}\).

    An algorithm is considered to provide a robust \(c\)-approximation for \(\ED\) if for all pairs \(X,Y\) it outputs \(k'\) such that: \(k'\in [k_{X,Y},ck_{X,Y}]\).
\end{definition}

We assert that for any value of \(\varepsilon\), any algorithm \(\ALG\) that provides a robust \(c\)-approximation for \(\ED\) yields an algorithm \(\ALG'\) that provides \((1+\varepsilon)c\)-approximation for \(\LCS\). Moreover, if the running time \(\ALG\) on input strings of lengths \(N,n\) is \(t(N,n)\), then the running time of \(\ALG'\) is \(t\left(\frac{n}{\varepsilon},n\right)\). The construction of $\ALG'$ is straightforward: on input strings \(X,Y\) we first apply the embedding,  then apply \(\ALG\) on the resulting strings and finally output: \(2k'\). 

We leave the quest of discovering a robust approximation algorithm for \(\ED\) as an open question, which falls outside the scope of this paper.

\subsection{Organization Of The Paper}
We structure the paper as follows. In Section~\ref{sec:alphabetReduction} we prove our main result, namely Theorem~\ref{thm:AlphabetReductionNormalized}, discussing normalized scaled isometric embedding between large and small alphabets. In Section~\ref{sec:AlphabetReductionBinary} we establish Theorem~\ref{thm:binaryAlphabet} focusing on alphabet reduction with binary alphabets. In Appendix~\ref{sec:IndelToEdit} we demonstrate the existence of an indel to edit scaled isometric embedding, as stated in Theorem~\ref{thm:IndelEditApx}.





\section{Preliminaries and Notations}\label{sec:preliminaries}
In this section we introduce the notations that is used throughout the rest of the paper.
For any string $X = x_1x_2\dots x_n$ and integers $i,j$, $X[i]$ denotes $x_i$ \footnote{We use $x_i$ or $X_i$ or $X[i]$ to denote the $i^{th}$ character of the string $X$ interchangeably.}, $X[i,j]$ represents substring $X' = x_i\dots x_j$ of $X$, and $X[i,j)=X[i,j-1]$.
``$\cdot$''-operator denotes concatenation, e.g $X\cdot Y$ is the concatenation of two strings $X$ and $Y$.
$\Lambda$ denotes the empty string.

\textbf{Edit Distance with Substitutions ($\ED$)}:
For strings $X,Y\in \Sigma^*$,  $\ED(X,Y)$ is defined as the minimal number of edit operations required to transform $X$ into $Y$. The set of edit operations includes character insertion, deletion, and substitutions.

\textbf {Indel Distance ($\LCS$)}:
For strings $X,Y\in \Sigma^*$,  $\LCS(X,Y)$ is defined as the LCS (Longest Common Subsequence) metric between $X$ and $Y$. It counts the minimal number of edit operations needed to convert the strings, where substitutions are excluded. 

\textbf{Normalized Distance:}
To assess the distance between each pair of strings in a standardized manner, it is advantageous to express it as a normalized value within the range \([0,1]\). To achieve this, we introduce the following definition:
\[\TED(X,Y)=\frac{\ED(X,Y)}{\max({\abs{X}, \abs{Y}})} , \, \,\,\,\, \TLCS(X,Y)=\frac{\LCS(X,Y)}{\abs{X}+ \abs{Y}}\]

A string $X'$ is considered a \textbf{subsequence} of another string $X$ if $\LCS(X,X') = |X|-|X'|$. 
$X'$ is considered a \textbf{substring} of $X$ if $X'$ is a contiguous subsequence of $X$. 

\textbf{Alignment}: 

For an alphabet $\Sigma$ and any two strings $X,Y\in \Sigma^*$ , an {\em $\ED$ alignment of $X$ and $Y$} is a a sequence of edit operations (insertions, deletions, and substitutions) that transform the string $X$ into $Y$. The cost of the alignment is determined by the number of edit operations. An alignment is optimal if it achieves  the lowest possible cost.  Observe that for each character of $X$ that wasn't deleted or substituted, can be matched with a unique character from $Y$. The collection of matched characters is referred to as the matching characters of the alignment. 

Similarly we define  an {\em $\LCS$ alignment of $X$ and $Y$}  as a sequence of edit operations, with the exception that substitutions are not permitted. The cost, optimality, and matching characters of an alignment are defined analogously. See Figure \ref{fig:alignment} for an example.

\begin{figure}[ht]
    \centering
    \includegraphics[width=0.6\textwidth]{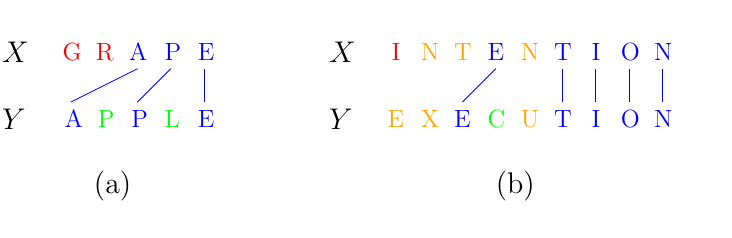}
    \caption{Example for (a) $\LCS$ alignment and (b) $\ED$ alignment (the matched characters are highlighted in blue, the deleted characters in red and the substituted characters in orange).}
    \label{fig:alignment}
\end{figure}

\section{Alphabet Reduction}\label{sec:alphabetReduction}
Within this section, we tackle the task of embedding strings from a sizable alphabet into a smaller one while preserving the global nature of the original metric space. Our main result demonstrates that it's possible to embed strings from any large alphabet \(\Gamma\) into strings of a smaller alphabet \(\Sigma\), where the length of the strings remains approximately unchanged, and the normalized distances are distorted by at most a factor of  \((1+\varepsilon)\). The size of the alphabet \(\Sigma\) increases quadratically with \(1/\varepsilon\).
To provide clarity, we present our results for the \(\LCS\) metric; a similar approach can be applied to the \(\ED\) metric. This section is structured as follows: In subsection~\ref{sec:normalizedEmbedding} we outline our findings regarding normalized scaled alphabet reductions, covering both lower and upper bounds. Section~\ref{sec:upperBounds} discusses our upper bounds, while Section~\ref{sec:lowerBounds} addresses lower bounds.

\subsection{Normalized Scaled Isometric Embedding - Our Finding}\label{sec:normalizedEmbedding}

An embedding $E$ is said to preserve lengths, if there exists: $\ell:\N \to \N$ such that: $\forall X \in \Gamma^n: \abs{E(X)}=\ell(n)$. 
It is non-shrinking if $\ell(n)\ge n$. 
It is natural to focus on non-shrinking embeddings otherwise if the embedding maps from a large alphabet to smaller one some distinct strings will get mapped to the same string. 

The following claim shows that any length-preserving embedding, mapping strings from large alphabet into strings of smaller one, necessarily contracts the normalized distances between certain pairs of strings. The proof of the claim is deferred to Section~\ref{sec:lowerBounds}.
\begin{claim}\label{claim:distortion}
   
    Let $\Gamma, \Sigma$ be finite alphabets, such that: $\abs{\Gamma}>\abs{\Sigma} $, and let $E:\Gamma^* \to \Sigma^*$, which is a length-preserving embedding, then for any $n\in \mathbb{N}$ we have:
\[ \exists X,Y \in \Gamma^n: \TLCS(E(X),E(Y))<\TLCS(X,Y). \]
\end{claim}


Therefore, we redirect our attention to approximate embedding, where we allow for a slight distortion in distances. Our main result is as follows:


\AlphabetReductionNormalized*

    


Note that the distortion is ``one-sided" in the sense that the normalized distances of the embedded strings cannot surpass the normalized distance of the original strings. However, for the lower bound, a \((1-\varepsilon)\)-factor may be incurred.  Furthermore, we demonstrate that for any embedding, the normalized distances cannot be uniformly scaled by a fixed factor. In particular, we demonstrate that there exist pairs of strings whose normalized distances are reduced, while for other pairs, their normalized distances converge to each other arbitrarily closely. This, in turn, illustrates that we cannot deduce the value of $\TLCS(X,Y)$ directly from the value of $\TLCS(E(X),E(Y))$ by a simple scaling.

\begin{claim}\label{claim:scaling}
        Let $\Gamma, \Sigma$ be finite alphabets, such that: $\abs{\Gamma}>\abs{\Sigma} $, and let $E:\Gamma^* \to \Sigma^*$, which is a length-preserving non-shrinking embedding.  We have:
        \begin{enumerate}
            \item For any $n\in \mathbb{N}$, there exist $X,Y \in \Gamma^n$ such that \[\TLCS(E(X),E(Y))\le(1-\frac{1}{|\Sigma|})\TLCS(X,Y).\] 
            \item For any sequence $Z_1,Z_2,\dots$ where $|Z_n|=n$, \[\lim_{n\rightarrow \infty}\TLCS(E(Z_n),E(\emptystr))-\TLCS(Z_n,\emptystr) = 0,\text {where \(\emptystr\) denotes the empty string.}\] 
         \end{enumerate}

\end{claim}

\subsection{Upper Bounds}\label{sec:upperBounds}

The crux of Theorem~\ref{thm:AlphabetReductionNormalized} lies in the existence of an error correcting code with respect to the $\LCS$ metric, even when the alphabet size is small. More specifically, given a proximity parameter $\varepsilon>0$, and $\abs{\Gamma}$, we pick a set $C$ of strings residing in $\Sigma^k$ of cardinality $\abs{\Gamma}$, with large pairwise distance. The construction of $C$ follows a greedy approach reminiscent of the Gilbert-Varshamov bound~\cite{Gilbert52,Varshamov57}. Properties of the code are summarized in the next statement.


\begin{lemma}\label{lem:code}
   For any $\varepsilon < 1/2$, let $\Sigma$ be a finite alphabet satisfying  $\abs{\Sigma} > 32/\varepsilon^2$. For every $n \in \mathbb{N}$, there exists $k \in \mathbb{N}$ with $k = O(\log n)$ for which the following conditions are satisfied:
   
    \begin{enumerate}
        \item There exists a code $C_{n,\varepsilon} \subseteq \Sigma^k$ with $\abs{C_{n,\varepsilon}} = n$.
        \item $\forall c\neq c' \in C_{n,\varepsilon}: \, \, \LCS(c,c')\ge (2-\varepsilon)k$.
    \end{enumerate}
\end{lemma}

The proof of Lemma \ref{lem:code} is given in Appendix \ref{appendix:code-alph-red}.

Endowed with the existence of such an error correcting code we assign a distinct codeword to each of the characters in the larger alphabet $\Gamma$. The embedding procedure is as follows: Given a string $X\in \Gamma^n$, we embed it into a string in $(\Sigma^k)^n$, where the encoding of $X$ is formed by concatenating the codewords assigned to each of its characters. The presentation of the embedding, along with its proof of correctness, is provided in Section~\ref{sec:embeddingAlphabetReduction}.

\subsubsection{The Embedding}\label{sec:embeddingAlphabetReduction}
Let $\Gamma$ be an alphabet, and let $\varepsilon$ be a proximity parameter, and let $C:=C_{\abs{\Gamma},\varepsilon}$ be the code whose existence is guaranteed by Lemma~\ref{lem:code}. We interpret the code $C$ as a function mapping characters from $\Gamma$ into (short) strings in $\Sigma^k$. The embedding proceeds as follows: each string in $\Gamma^*$ is encoded sequentially character by character, where the encoding of each character is performed using the code $C$. Our main technical lemma is as follows:

\begin{lemma}\label{lem:embedding}
   For any finite alphabet $\Gamma$ and $\varepsilon > 0$, consider the encoding $C_{\abs{\Gamma},\varepsilon}:\Gamma\to \Sigma^k$ as implied by Lemma~\ref{lem:code}. For any string $X\in \Gamma^*$ define: $E(X)=C(X_1)\dots C(X_n)$ (where $n=\abs{X}$). 
   
   Then for any pair of strings $X,Y\in \Gamma^*$ the following inequality holds:
   \[ (1-48\varepsilon)\LCS(X,Y) \le \LCS(E(X),E(Y)) \le \LCS(X,Y)\]
   Moreover, for every $X\in \Gamma^n$ we have: $\abs{E(X)}=O(n\log \abs{\Gamma})$.
\end{lemma}


In the sequel, an \(\LCS\)-alignment converting \(E(X)\) into \(E(Y)\) is simply referred as an alignment. In the course of the proof, we introduce the concepts of blocks and block-structured alignments. For $X\in \Gamma^n$ we define the $i$-th block of $E(X)$ to be the substring of $E(X)$ corresponding to $X_i$, namely it equals $C(X_i)$.
Furthermore, given an alignment $\A$ that transforms $E(X)$ into $E(Y)$, we label it as a {\em block-structured} alignment if, for each $i$-th block in $E(X)$, the alignment either fully matches all the characters of the block to some block $j$ in $E(Y)$ or entirely deletes the $i$-th block.
It is clear that block-structured alignments for the embedded strings correspond one-to-one with alignments for the original strings, and their normalized distance remains unchanged.
To prove our main technical lemma we transform any alignment converting \(E(X)\) to \(E(Y)\) into a block-structured one without significantly increasing its cost. 

We will introduce certain notations to facilitate the presentation of the proof. For any $X\in \Gamma^n$ we employ lowercase letters such as: $i,j,k$ etc. to represent indices of $X_i$. We utilize tuples from $[n]\times [k]$ to represent indices of $E(X)$, where the first index signifies the block index denoted by lowercase letters, and the second one describes the index within the block represented by a lowercase Greek letter.

The following claim, stated without a formal proof, will be useful in the subsequent proof.
\begin{claim}\label{claim:alignEmbedding}
    For an alignment $\A$ transforming $E(X)$ into $E(Y)$, we have that the set of matching characters has to be monotone,  indicating  that for $(i,\alpha)<(i',\alpha')$ in lexicographical order if $(i,\alpha)$ is matched  to $(j,\beta)$ and $(i',\alpha')$ to $(j',\beta')$ by $\A$, we must have: $(j,\beta)<(j',\beta')$.

\end{claim}

Given an alignment $\A$, we partition $E(Y)$ into $n$ segments based on its matching blocks in $E(X)$. Define the $i$-th segment as follows: if no character of the $i$-th block of $E(X)$ is matched under $\A$, then the $i$-th segment is empty. Otherwise, let $j$ denote the first block of $E(Y)$ that includes a matching character for one of the characters in the $i$-th block. If $i$ is the smallest block containing a matching coordinate within $j$, then the $i$-th segment starts at $(j,1)$, otherwise it starts at the first coordinate within the $j$-th block matching with the $i$-th block.  

The ending point of the segment is defined similarly: Let $j'$ be the initial block of $E(Y)$ containing a match for the $(i+1)$-th block of $E(X)$. If the $i$-th block does not match any of the $j'$-th coordinates, then the $i$-th segment ends at $(j-1,k)$. Otherwise, it ends at the coordinate preceding the first match of $(i+1)$ and $j$. We define the starting point of the first non-empty segment as $(1,1)$ and the end point of the last non-empty segment as $(n,k)$. See Figure~\ref{fig:align} (in appendix) for an illustration.



Additionally, we define $cost_{\A}(i)$, the cost of the $i$-th block, as the sum of unmatched coordinates in $E(X)$ within its $i$-th block and the unmatched coordinates in $E(Y)$ within its $i$-th segment. For example, in the illustration provided in Figure~\ref{fig:align} (in appendix), we have $cost_{\A}(1)=0+3$ since all the characters of the first block of $E(X)$ are matched, and there are $3$ unmatched coordinates in the first segment of $E(Y)$. As the decomposition of $E(Y)$ results in disjoint parts, the sum of the costs across the different blocks equals the cost of $\A$, which we denote by: $cost(\A)$.

\paragraph{Converting Into Block-Structured Alignments}
In this section, we introduce an algorithm that takes an arbitrary alignment $\A$ transforming $E(X)$ into $E(Y)$ as input and produces an alignment $\A^*$. 
The resulting alignment $\A^*$ is block-structured, and its cost does not substantially exceed that of $\A$. 

To design the new matching we need few definitions. We say that blocks $i$ and $j$ are \textit{partially matched} by an alignment $\A$ if there exists a pair $(i,\alpha)$ and $(j,\beta)$ matched by $\A$. 
Furthermore, $i$ and $j$ are \textit{significantly matched} by $\A$ if more than $\varepsilon k$ characters in the $i$-th block of $X$ are matched into the $j$-th block; 
we say that $i$ and $j$ are \textit{perfectly matched} if $\A$ matches every character in $E(X)$ into a character in $E(Y)$.
The algorithm operates in two stages. 
In the first stage, the algorithm iteratively takes two significantly matched blocks that are not perfectly matched, 
removes all matches that the characters of the two blocks participate in,
and introduces a perfect match between the two blocks.
In the second stage, we remove all the matches from blocks that are not matched perfectly. 

A key observation is that if $i$ and $j$ are significantly matched then we have the following inequality: $\LCS(C(X_i),C(Y_j))<(2-\varepsilon)k$. 
Therefore, by the distance guarantee regarding $C$, we must have $X_i=Y_j$.  
The algorithm is given next. 
For the sake of the analysis its second stage is divided into two cases.

\begin{algorithm}[ht]
\caption{Converting Into Block-Structured Alignments: }\label{alg:block-Structured}
\KwData{An alignment $\A$ converting $E(X)$ into $E(Y)$}
\KwResult{An alignment $\A'$ that is block-structured, converting $E(X)$ into $E(Y)$}
$\A'\gets A$;

\Comment*[r]{Stage I:}
\For{$i=1 \cdots \abs{X}$}{

  \If{$i$ has some significant match}{
    $j \gets \text{smallest block in $E(Y)$ that significantly matches $i$}$;
    
    Remove all matches incident with blocks $i$ and $j$ from $\A'$, and add a \textit{perfect} match between the two blocks;
    
}
}
$i \gets 1$ \Comment*[r]{Stage II:}
    
\While{$i<= \abs{X}$}{
  \If{$i$ has a partial and not perfect match}{
        
        \If{$i$ is partially matched to more than a single block} 
        {
        
            Delete from $\A'$ all matches of characters from the $i$-th block of $E(X)$; 
            
            $i++$;
        }
        \Else{
                       
            $j \gets \text {smallest $E(Y)$-block that partially matches $i$}$; 

            $i'\gets$ smallest $E(X)$ that does not match with the $j$-th block;

            Delete from $\A'$ all matches of characters from the $j$-th block of $E(Y)$;
            
            $i\gets i'$;
        }

    
    
}
}
\end{algorithm}

\paragraph{The Correctness of the Algorithm}
We break down the proof of correctness into three claims. Claim~\ref{claim:blockStructure} states that the output produced by Algorithm~\ref{alg:block-Structured} is a block-structured alignment. The subsequent two claims provide bounds on the cost difference between the input and output alignments.

\begin{claim}\label{claim:blockStructure}
    The alignment $\A'$ produced by Algorithm~\ref{alg:block-Structured} from any alignment $\A$ converting $E(X)$ into $E(Y)$ is a block-structured alignment converting $E(X)$ into $E(Y)$.
\end{claim}


The proof of Lemma~\ref{lem:embedding} is derived from the following claims, let us first state the claims.


\begin{claim}\label{claim:stageI}
    Let $\varepsilon<1/4$ and let $\A$ be any alignment converting $E(X)$ into $E(Y)$. Let $\A_{\mathrm{I}}$ be the resulting alignment obtained by applying the algorithm described in stage I on $\A$ with a proximity parameter of $\varepsilon$.  Then, $$cost(\A_{\mathrm{I}})\le (1+4\varepsilon)cost(\A).$$
\end{claim}

\begin{claim}\label{claim:stageII}
    Let $\A_{\mathrm{I}}$ be any resulting alignment obtained by applying the algorithm described in stage I on some alignment $\A$ with a proximity parameter of $\varepsilon$. Let $\A_{\mathrm{II}}$ be the resulting alignment obtained by applying the algorithm described in stage II on $\A_{\mathrm{I}}$ with a proximity parameter of $\varepsilon$. Then, 
    $$cost(\A_{\mathrm{II}})\le (1+4\varepsilon)cost(\A_{\mathrm{I}}).$$
\end{claim}

The proofs of claims \ref{claim:blockStructure}, \ref{claim:stageI} and \ref{claim:stageII} are given in Appendix \ref{appendix:code-alph-red}.

\begin{proofof} {Lemma~\ref{lem:embedding}}
[using Claim~\ref{claim:stageI} and Claim~\ref{claim:stageII}]


Let $OPT$ represent the normalized cost of the optimal alignment between $X$ and $Y$, and $\widetilde{OPT}$ denote the normalized cost of the optimal alignment between $E(X)$ and $E(Y)$.
Notice that any alignment between $X$ and $Y$ can be paired with an alignment between $E(X)$ and $E(Y)$ having the same cost. Consequently, we have: \(\widetilde{OPT}\le OPT.\)
To complete the argument, it remains to establish that: \((1-48\varepsilon) OPT \leq \widetilde{OPT}\),  which can be achieved by demonstrating: \(OPT \leq (1+24\varepsilon)\widetilde{OPT}\).

Consider the optimal alignment $\A$ that transforms $X$ into $Y$, and let $\A'$ be the alignment generated by Algorithm~\ref{alg:block-Structured} when applied to $\A$. According to Claims~\ref{claim:stageI} and~\ref{claim:stageII}, we obtain:
\[\frac{1}{2|E(X)|}\cdot cost(\A') \le \frac{1}{2|E(X)|}\cdot (1+4\varepsilon)^2 cost(\A) \le \frac{1}{2|E(X)|}\cdot (1+24\varepsilon) cost(\A) = (1+24\varepsilon)OPT. \]
We conclude the proof by noting that: \( \widetilde{OPT}\le \frac{1}{2|E(X)|} \cdot cost(\A')\).
\end{proofof}



\subsection{Lower Bounds}\label{sec:lowerBounds}

\begin{proofof}{Claim~\ref{claim:distortion}}

For any value of $n\in \N$, define $A_n$ as the set of length $n$ strings composed of a single character from $\Gamma$. Clearly, $\abs {A_n}=\abs{\Gamma}$ and moreover, for every distinct pair of strings in $A_n$, their Indel distance is $2n$.

Now consider any embedding $E:\Gamma^*\to \Sigma^*$. Since $\abs {A_n}=\abs{\Gamma}>\abs{\Sigma}$,  by the pigeonhole principle there exist $X\neq Y\in A_n$ satisfying: $E(X)_1=E(Y)_1$ \footnote{$E(X)_i$ is the $i^{th}$ character of the string $E(X)$.}. Hence, $\LCS(E(X),E(Y))<2\ell(n)$ while $\ED(X,Y)=2n$. 
\end{proofof}

\begin{proofof}{Claim~\ref{claim:scaling}}
\begin{enumerate}
    \item Fix $n\in \N$ and consider any embedding $E:\Gamma^*\to \Sigma^*$. For any $X\in \Gamma^n$, define the value $p(E(X))\in \Sigma$ as the plurality value among $\{E(X)_i\}_{i\in \N}$ (ties are broken arbitrarily). Observe that the character $p(E(X))$ appears at least $\frac{\ell(n)}{\abs{\Sigma}}$ times in the string $E(X)$. Furthermore, for any $X,Y\in \Sigma^n$ if: $p(E(X))=p(E(Y))$, then we get: $LCS(E(X),E(Y))\ge \frac{\ell(n)}{\abs{\Sigma}}$ and hence: $\LCS(E(X),E(Y))\le \left(1-\frac{1}{\abs{\Sigma}}\right)2\ell(n)$.

As in the proof of Claim~\ref{claim:distortion}, define $A_n$ as the set of strings composed of a single character from $\Gamma$. Recall that $\abs {A_n}=\abs{\Gamma}$ and moreover, for every distinct pair of points in $A_n$, their $LCS$ distance is $2n$.

Since $\abs {A_n}=\abs{\Gamma}>\abs{\Sigma}$,  by the pigeonhole principle there exist $X\neq Y\in A_n$ satisfying: $p(E(X))=p(E(Y))$, yielding: $\LCS(E(X),E(Y))\le \left(1-\frac{1}{\abs{\Sigma}}\right)2\ell(n)$, whereas $\LCS(X,Y)=2n$,  as claimed.
\item Let $k=|E(\emptystr)|$. For $Z\in \Gamma^n$, we have: $\LCS(E(Z),E(\emptystr))\ge \ell(n)-k$ so  $\TLCS(E(Z),E(\emptystr))\ge 1-\frac{k}{\ell(n)} \ge 1 - \frac{k}{n} $. On the other hand, $\LCS(Z,\emptystr)=n$ so $\TLCS(Z,\emptystr)=1$.
\end{enumerate}\end{proofof}

\section{Alphabet Reduction - Binary Alphabets}\label{sec:AlphabetReductionBinary}

In this section we show a reduction of $\ED$ and $\LCS$ over an arbitrary alphabet to the binary alphabet.
The reduction expands the strings super-polynomially, but one can think of it as a proof of concept that more efficient reduction might exist. 
The main theorem of this section is the following statement which is a formal statement
of Theorem \ref{thm:binaryAlphabet}.
For ease of presentation it is beneficial to think about Longest Common Subsequence instead of $\LCS$.
That is how we state the theorem here.

 \begin{theorem}\label{thm:formulas}
 For any integer $n\ge 1$, any alphabet $\Sigma$ of size at most $n^3$, there exist integers $S,R,N$ where $N=n^{O(\log n)}$
 and functions $G,H, G', H' : \Sigma^n \rightarrow \{0,1\}^N$
 such that  for any $X,Y \in \Sigma^n$,
 \begin{eqnarray*}
 LCS(X,Y) & = & \frac{LCS(G(X),H(Y)) - R}{S} \\
 \ED(X,Y) & = & \frac{LCS(G'(X),H'(Y)) - R}{S} . 
 \end{eqnarray*}
 \end{theorem}
Hence, for any pair of strings $X,Y$ one can recover $\ED(X,Y)$ from $\LCS(G'(X),H'(Y))$ over a binary alphabet.
Both mappings $G,H$ and $G',H'$ can be computed efficiently in the length of their output. 
Indeed, they will be defined explicitly below.
We remark that the bound $n^3$ on the size of $\Sigma$ is essentially arbitrary and could be replaced for example by a bound $2^n$
without change in the other parameters (except for multiplicative constants).
However, the $n^3$ bound allows for hashing any large alphabet by a random pair-wise independent hash function to an alphabet of size $n^3$
without affecting the distance of any given pair of strings except with probability $<1/n$.

In order to prove the theorem we will need several auxiliary functions.
We say that a $0$-$1$ string is {\em balanced} if it contains the same number of 0's and 1's. 
We say a formula $\phi$ is {\em normalized} if it consists of alternating layers of binary $AND$ and $OR$ and all of its literals are at the same depth; 
each literal is either a constant, a variable or its negation.

We define two functions $g,h : \{0,1\}^* \times \{{\mathrm{normalized\  formulas}}\} \rightarrow \{0,1\}^*$ and two threshold functions $f,t:  \{{\mathrm{normalized\  formulas}}\} \rightarrow \mathbb{N}$ as follows:
Let us consider sets of variables $U=\{u_1, \dots , u_p\}$ and $V = \{v_1,\dots , v_q\}$, and let $A = \{a_1, \dots , a_p\}$ where $a_i$ is the assignment to the variable $u_i$ for all $1\le i\le p$, and $B = \{b_1, \dots , b_q\}$ where $b_i$ is the assignment to the variable $v_i$ for all $1\le i\le q$.

Let \(\phi(U,V)\) be a normalized formula which is defined over two disjoint sets of variables $U=\{u_1, \dots , u_p\}$ and $V = \{v_1,\dots , v_q\}$. Let $A \in \{0,1\}^p ,B\in \{0,1\}^q$  where $A$ and $B$ are interpreted as assignments for $U$ and $V$ respectively. 
We define  two functions $g,h$, such that $g$ gets as an input a pair $(\phi,A)$ and outputs a string in $\{ 0,1\}^*$, similarly $h$ takes a pair $(\phi,B)$ as its input and outputs a string in $\{ 0,1\}^*$. We also define threshold functions $f,t$ which take such a formula as input and output a natural number. The crux of the construction is that for any assignment $A$ for $U$ and $B$ for $V$ we have that if $\phi$ is satisfied by the assignment pair $A,B$ then $LCS(g(\phi,A),h(\phi,B))=t(\phi)$, otherwise $LCS(g(\phi,A),h(\phi,B))=f(\phi)$

We establish the recursive definitions of $g,h,f$ and $t$  based on the depth of the formula. The base case is when $\phi$ is either a constant $0,1$ or single literals $u_i, \neg u_i, v_j , \neg v_j$, where, $u_i\in U , v_j\in V$. Here by $\neg 0$ we understand symbol $1$, and similarly by $\neg 1$ we understand symbol $0$.

\noindent\bigskip
\begingroup
\renewcommand{\arraystretch}{1.25} 
\begin{center}
\begin{tabular}{c|c|c|c|c|c|c}
              &  $\phi = u_i$    & $\phi = \neg u_i$ & $\phi = v_j$    & $\phi = \neg v_j$ & $\phi = 1 $ & $\phi = 0 $  \\
\hline
 $g(A,\phi)$  &  $\neg a_i a_i$  & $a_i \neg a_i$    & 0 1             &  0 1              &   0 1       &  1 0\\
 $h(B,\phi)$  &  0 1             &  0 1              & $\neg b_j b_j$  & $b_j \neg b_j$    &   0 1       &  0 1\\
\hline
 $t(\phi)$    &  2               &  2                & 2               & 2                 & 2               & 2  \\
 $f(\phi)$    &  1               &  1                & 1               & 1                 & 1               & 1   
\end{tabular}
\end{center}
\endgroup

\bigskip\noindent

and further inductively:

\bigskip\noindent
\begingroup
\renewcommand{\arraystretch}{1.5} 
\begin{tabular}{c|c|c}
              &  $\phi = \phi_0 \;OR\; \phi_1 $ & $\phi = \phi_0 \;AND\; \phi_1$   \\
\hline
 $g(A,\phi)$  &  $1^{k/2}\,\,1^{4k}\,\,g(A,\phi_0)\,\,1^{4k}\,\,0^{4k}\,\,g(A,\phi_1)\,\,0^{4k}\,\,0^{k/2}$  & $0^{T+F}\,\,1^{11k+T+F}\,\,0^{5k}\,\,g(A,\phi_0)\,\,0^k\,\,1^k\,\,0^k\,\,g(A,\phi_1)\,\,0^{5k}$  \\
 $h(B,\phi)$  &  $0^{k/2}\,\,0^{4k}\,\,h(B,\phi_0)\,\,0^{4k}\,\,1^{4k}\,\,h(B,\phi_1)\,\,1^{4k}\,\,1^{k/2}$ & $0^{T+F}\,\,0^{5k}\,\,h(B,\phi_0)\,\,0^k\,\,1^k\,\,0^k\,\,h(B,\phi_1)\,\,0^{5k}\,\,1^{11k+T+F}$    \\
\hline
 $t(\phi)$    &  $9k+T$  &  $13k+3T+F$   \\
 $f(\phi)$    &  $9k+F$  &  $13k+2T+2F$.
\end{tabular}
\endgroup
\bigskip\noindent
where  $k=|g(x,\phi_0)|$, $T=t(\phi_0)$, and $F = f(\phi_0)$.

Key properties of our functions are summarized in the next lemma.

\begin{lemma}\label{lem:fleeval}
Let $\phi(U,V)$ be a balanced formula of depth $d$ with set of variables $U = \{u_1, \dots , u_p\}$ and $V = \{v_1, \dots , v_q \}$. 
For every two assignments $A,A'\in\{0,1\}^p$ to variables $U$, 
we have $|g(A,\phi)|=|g(A',\phi)|$. Similarly, for every two assignments $B,B'\in\{0,1\}^q$ to variables $V$, $|h(B,\phi)|=|h(B',\phi)|$.
Additionally, $|g(A,\phi)|=|h(B,\phi)| \le 30^d$.

Furthermore, the following holds:
\begin{itemize}
    \item If $\phi(A,B)$ is true then $LCS(g(A,\phi),h(B,\phi))=t(\phi)$.
    \item If $\phi(A,B)$ is false then $LCS(g(A,\phi),h(B,\phi))=f(\phi)$.
\end{itemize}
Finally, $f(\phi) < t(\phi)$.
\end{lemma}

In order to prove the above lemma we also need two gadgets which we call the $AND$-gadget and the $OR$-gadget. We need the lemmas on these gadgets (statement and proofs included in Appendix \ref{appendix:LCS-formula}) which analyze the composition of $AND$ and $OR$.

The proofs of theorem \ref{thm:formulas} and lemma \ref{lem:fleeval} can also be found in Appendix \ref{appendix:LCS-formula}.

\newpage
\bibliographystyle{plainurl}
\bibliography{references}

\newpage
\appendix
\section*{Appendix}

\section{Alphabet Reduction}\label{appendix:code-alph-red}
We use the following simple fact to prove Lemma~\ref{lem:code}:

\begin{proposition}\label{prop:varepsilon}
For any integer $k\ge 1$ and any real $0<\varepsilon<1/2$, $\binom {k} {\lfloor \varepsilon k \rfloor} \le 2^{(\varepsilon \log(1/\varepsilon) + 2 \varepsilon)k}$.
\end{proposition}

\begin{proofof}{Proposition~\ref{prop:varepsilon}}
    Let $\varepsilon' = \lfloor \varepsilon k \rfloor / k$ so $\varepsilon' \le \varepsilon$.
    It is well known that $\binom {k} {\varepsilon' k} \le 2^{H(\varepsilon') k}$, where $H(x)=x\log (1/x) + (1-x)\log 1/(1-x)$ is the binary entropy function.
    So $\binom {k} {\lfloor \varepsilon k \rfloor} = \binom {k} {\varepsilon' k} \le 2^{H(\varepsilon') k} \le 2^{H(\varepsilon) k}$, as $H(\cdot)$ is increasing on the interval $[0,1/2]$.
    Notice, $\log 1/(1-\varepsilon) = \log (1+\frac{\varepsilon}{1-\varepsilon}) \le \log e^{\frac{\varepsilon}{1-\varepsilon}} = \frac{\varepsilon}{1-\varepsilon} \cdot \log e$, where we use the inequality $1+x \le e^x$ for any real $x$.
    Hence, $H(\varepsilon) \le \varepsilon \log (1/\varepsilon) + \varepsilon \cdot \log e \le \varepsilon \log (1/\varepsilon) + 2 \varepsilon$.
\end{proofof}

\begin{proofof}{Lemma~\ref{lem:code}}

The proof employs the probabilistic method, demonstrating that there exists a random set $C \subseteq \Sigma^k$ meeting the specified criteria with a non-zero probability, thereby establishing its existence. The parameters of the constructions are as follows:
pick an integer $k$ from the interval $[\frac{2}{\varepsilon} \log n, \frac{1}{\varepsilon}+\frac{2}{\varepsilon} \log n]$ so that $\varepsilon k$ is an integer  and let $\abs{\Sigma}=\lceil 32/\varepsilon^2 \rceil$.

Our initial insight is that, given any $X,Y \in \Sigma^k$, if for all subsets $I_1, I_2 \subset [\ell]$ where $\abs{I_1}=\abs{I_2}=\varepsilon k$, it holds that: \( X_{I_1} \neq Y_{I_2}\) then \(LCS(X,Y)<\varepsilon k\) and \(\LCS(X,Y)\ge (2-2\varepsilon)k\).

Next, for $X,Y$ chosen uniformly at random, and for a fixed choice of $I_1, I_2 \subset [\ell]$ where $\abs{I_1}=\abs{I_2}=\varepsilon k$ we have that: \[ \Pr[X_{I_1}=Y_{I_2}]=\abs{\Sigma}^{-\varepsilon k}\]

Applying a union bound to all possible choices of $I_1,I_2$ we obtain: 
\[ \Pr[\forall I_1, I_2: X_{I_1}=Y_{I_2}]=\abs{\Sigma}^{-\varepsilon k}{\binom {k } {\varepsilon k}}^2 \le \abs{\Sigma}^{-\varepsilon k} 2^{2 (\varepsilon \log(1/\varepsilon) + 2 \varepsilon)k} \le 2^{-\varepsilon k},\]
where the final inequality is a consequence of our choice of $\Sigma$ where $\log \abs{\Sigma} \ge  2\log(1/\varepsilon) + 5$. 

In summary, with a probability of at most $2^{-\varepsilon k}$, it follows that for uniformly random $X$ and $Y$ in $\Sigma^k$, their LCS surpasses $\varepsilon k$.
Ultimately, consider a set $C\subseteq \Sigma^k$ with $\abs{C}=n$. Applying a union bound to all pairs within $C$, we deduce that the probability of the existence of a pair $c\neq c'$ such that $LCS(c,c')>\varepsilon k$ is at most $\binom{\abs C}{2} 2^{-\varepsilon k}<1$. The last inequality arises from our choice of $k \ge \frac{2}{\varepsilon} \log n$, and this establishes the existence of $C$, as claimed.
\end{proofof}

\begin{figure}[ht]
			\begin{tikzpicture}[scale=1.5,shorten >=1mm,>=latex]
			\tikzstyle gridlines=[color=black!20,very thin]
		
			\node at (-2,0.3) {$E(X)$};
			\node at (-2,-1.25) {$E(Y)$};
                \node at (-2,-0.5) {$\A$};
		
                \node at (2.0,0.65) {$C(X_1)$};
                \draw[] (1.5,0) rectangle (2.6,0.5);
			\node at (1.75,0.25) {$1$};
                \node at (2.0,0.25) {$1$};
                \node at (2.25,0.25) {$0$};

                \node at (3.1,0.65) {$C(X_2)$};
                \draw[] (2.6,0) rectangle (3.7,0.5);
			\node at (2.85,0.25) {$0$};
                \node at (3.1,0.25) {$1$};
                \node at (3.35,0.25) {$0$};

                \node at (4.0,0.65) {$C(X_3)$};
                \draw[] (3.7,0) rectangle (4.8,0.5);
			\node at (4.0,0.25) {$1$};
                \node at (4.25,0.25) {$1$};
                \node at (4.5,0.25) {$1$};
                
                \node at (1.4,-1.65) {$C(Y_1)$};
                \draw[] (0.95,-1.5) rectangle (2.05,-1);
                \node at (1.2,-1.25) {$0$};
                \node at (1.45,-1.25) {$0$};
                \node at (1.7,-1.25) {$1$};

                \node at (2.5,-1.65) {$C(Y_2)$};
                \draw[] (2.05,-1.5) rectangle (3.15,-1);
                \node at (2.3,-1.25) {$1$};
                \node at (2.55,-1.25) {$0$};
                \node at (2.8,-1.25) {$1$};

                \node at (3.6,-1.65) {$C(Y_3)$};
                \draw[] (3.15,-1.5) rectangle (4.25,-1);
                \node at (3.4,-1.25) {$0$};
                \node at (3.65,-1.25) {$0$};
                \node at (3.9,-1.25) {$1$};

                \node at (4.7,-1.65) {$C(Y_4)$};
                \draw[] (4.25,-1.5) rectangle (5.35,-1);
                \node at (4.5,-1.25) {$0$};
                \node at (4.75,-1.25) {$1$};
                \node at (5.0,-1.25) {$0$};

			\draw[->,thick] (1.75,0)--(1.7,-1);
                \draw[->,thick] (2.0,0)--(2.3,-1);
                \draw[->,thick] (2.25,0)--(2.55,-1);
                \draw[->,thick] (2.85,0)--(3.65,-1);
                \draw[->,thick] (3.1,0)--(3.9,-1);
                \draw[->,thick] (3.35,0)--(4.5,-1);
                \draw[->,thick] (3.35,0)--(4.5,-1);
                \draw[->,thick] (4.5,0)--(4.75,-1);
		      
                \draw[dashed] (1.5,0)--(0.95,-1);
                \draw[dashed] (2.6,0)--(3.15,-1);
                \draw[dashed] (3.7,0)--(4.6,-1);
                \draw[dashed] (4.8,0)--(5.35,-1);
                
			\end{tikzpicture}

			\caption{An illustration of the matching between the strings $E(X)$ and $E(Y)$. Arrows indicate matching coordinates, and dashed lines represent the beginning/end points of the segments. The first segment starts at $(1,1)$ and ends at $(2,3)$, as the first matched coordinate in the second block of $E(X)$ is mapped to the third block, and no character from the first block of $E(X)$ is mapped to that block. The second segment starts at $(3,1)$ and ends at $(4,1)$ (as the first matched coordinate in the second block of $E(X)$ is mapped to the third block, and there exists a character from the second block of $E(X)$ that is mapped to that block)
   }
   
\label{fig:align}

\end{figure}
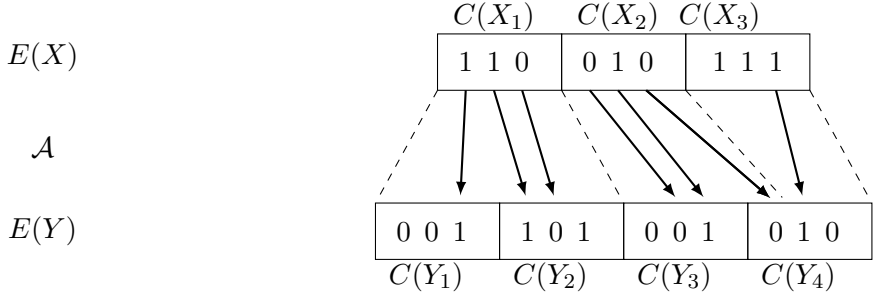
\begin{proofof}{Claim~\ref{claim:blockStructure}}
We first prove that $\A'$ converts $E(X)$ into $E(Y)$. Observe that $\A$ converts $E(X)$ into $E(Y)$, and $\A'$ is derived by alternately performing a perfect match for blocks with a significant match and removing other matched characters from $\A$, it is evident that the resulting alignment, $\A'$, indeed converts $E(X)$ into $E(Y)$.

Next we claim that the only matched characters are within blocks which are perfectly matched.  Consider each block $i$: If $i$ has a significant match under $\A$, then under $\A'$, it is matched perfectly to a single block $j$. Next, we need to demonstrate that in $\A'$, no partial matching exists.

For blocks $i$ that are partially matched under $\A$ into a \textit{single} block $j$, or vice versa, it is straightforward to verify the absence of partial matching in $\A'$.
However, consider a block $i$ that is matched into several blocks under $\A$, and one of these blocks has several partial matches as well. In this scenario, due to the monotonicity property (as per Claim~\ref{claim:alignEmbedding}), it is either the case that $i$ is matched to several blocks, and then only the last matched block $j$ is matched to several blocks, or vice versa.

Let us analyze the first case; the second one follows by the same reasoning. In this case, the algorithm first removes all the matches involving the coordinates of the $i$-th block during the $i$-th iteration. Subsequently, in the following iteration, it removes all the matches involving the coordinates of the $j$-th block. In total, this results in a block-structured alignment.    
\end{proofof}

\begin{proofof}{Claim~\ref{claim:stageI}}
    We show that for each block $i\in [n]$, the cost of the $i$-th block under $\A$ and $\A_{\mathrm{I}}$ does not change substantially. Specifically, as we scan the blocks from left to right, we establish that for each block, $cost_{\A_{\mathrm{I}}}(i)\le (1+4\varepsilon)cost_{\A}(i)$. The proof of Claim~\ref{claim:stageI} follows by summing the cost differences over the various blocks.
    
    Consider $i\in [n]$ and assume that the $i$-th block of $E(X)$ significantly matches the $j$-th block of $E(Y)$, with $j$ being the smallest one that satisfies this condition. In stage I, the algorithm opts to establish a perfect match between $i$ and $j$. This may impact (i) the cost of the $i$-th block, (ii) the cost of any block $i'>i$ that has partial match with the $j$-th block of $E(Y)$ and (iii) the cost of any block $i'<i$ that has partial match with the $j$-th block of $E(Y)$. 

    As for the $i$-th block, we affirm that its cost under $\A_{\mathrm{I}}$ can only decrease. Given that the $j$-th block of $E(Y)$ significantly matched with $i$, and as observed earlier, we have $X_i=Y_j$. During phase I, all the characters in $C(X)_i$ are matched to the characters in $C(Y)_j$, so there is no increase in the number of deleted characters in the $i$-th block. However, in $\A_{\mathrm{I}}$, the characters in blocks $j'>j$, which were previously matched under $\A$ to characters in the $i$-th block of $E(X)$, are deleted and no longer matched  under $\A_{\mathrm{I}}$. Nevertheless, each deletion can be accounted for by each new matching in the $j$-th block. 

    Consider blocks $i'>i$ that were partially matched to the $j$-th block of $E(Y)$. Under $\A_{\mathrm{I}}$ all these matches are deleted. Nevertheless, each such a deletion can be acounted against a charachter of the $i$-th block of $E(X)$ that is deleted by $\A$ and is matched under $\A_{\mathrm{I}}$.

     It is left to consider any block $i'<i$ that has a partial match with the $j$-th block of $E(Y)$.  We claim that the $i'$-th block lacks a significant match with any of the blocks in $E(Y)$. For the $j$-th block, this follows from the description of phase I. Regarding previous blocks, if a significant match existed, then in prior iterations, it would have been perfectly mapped to a preceding block, and consequently, it would not have any matching with the $j$-th block.
    
    We next claim that the cost of $i'$-th block is at least $(1-2\varepsilon)k$. The proof is conducted through case analysis based on the number of blocks in $E(Y)$ that were partially matched into this block:

   If the characters of the block were mapped solely to the $j$-th block, then in $\A$, at least $(1-\varepsilon)k$ of its characters are deleted, and the proof follows. If it was mapped to one additional block, then under $\A$, at least $(1-2\varepsilon)k$ of its characters are deleted, and the proof follows. The remaining case to analyze is when it was mapped to at least three blocks; denote the first three of them by $j_1< j_2< j_3$. 
   
   Consider the block $j_2$, we claim that at least $(1-\varepsilon)k$ of its characters are deleted by $\A$. Since it does not have a significant match with $i'$, at most $\varepsilon k$ of its characters are matched into the $i'$-th block of $E(X)$ by $\A$. We claim that all the others are deleted. Since some characters of the $i'$-th block are matched with $j_1$, there is no matching between characters from $j_2$ to any block preceding $i'$ otherwise the matching is not monotone, see Figure~\ref{fig:multipleMatching} for illustration. By similar reasoning these characters cannot be matched to block which following $i'$. Hence the characters which are not matched to $i'$ are deleted, as claimed. In total we conclude that: $cost_{\A}(i')\ge (1-\varepsilon)k$.

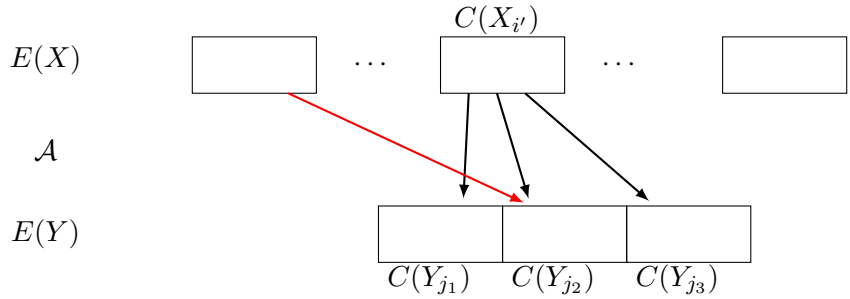
\begin{figure}[ht]
			\begin{tikzpicture}[scale=1.5,shorten >=1mm,>=latex]
			\tikzstyle gridlines=[color=black!20,very thin]
		
			\node at (-2,0.3) {$E(X)$};
			\node at (-2,-1.25) {$E(Y)$};
                \node at (-2,-0.5) {$\A$};
		
                \node at (2.0,0.65) {$C(X_{i'})$};
                \draw[] (-0.7,0) rectangle (0.4,0.5);
                \node at (0.9,0.25) {$\cdots$};
                \node at (3.1,0.25) {$\cdots$};
                \draw[] (1.5,0) rectangle (2.6,0.5);

                \draw[] (4.0,0) rectangle (5.1,0.5);

                \node at (1.4,-1.65) {$C(Y_{j_1})$};
                \draw[] (0.95,-1.5) rectangle (2.05,-1);

                \node at (2.5,-1.65) {$C(Y_{j_2})$};
                \draw[] (2.05,-1.5) rectangle (3.15,-1);

                \node at (3.6,-1.65) {$C(Y_{j_3})$};
                \draw[] (3.15,-1.5) rectangle (4.25,-1);
			
			\draw[->,thick] (1.75,0)--(1.7,-1);
                \draw[->,thick] (2.0,0)--(2.3,-1);
                \draw[->,thick] (2.25,0)--(3.4,-1);
                \draw[->,thick,red] (0.15,0)--(2.3,-1);
		
			\end{tikzpicture}

			\caption{Block $i'$ partially matches blocks $j_1, j_2, j_3$ in $E(Y)$. Consequently, no other block in $E(X)$ can be partially matched with $j_2$. The red line illustrates the prohibited matching. }
   
\label{fig:multipleMatching}
\end{figure}

    In either case, the cost of the $i'$-th block is at least $(1-2\varepsilon)k$, as claimed. Observe that during the $i$-th iteration at most $\varepsilon$ characters of the $i'$-th block  were deleted and this quantity were deleted in the $j$-th block of $E(Y)$ we get: 
    $$cost_{\A_{\mathrm{I}}}(i')\le cost_{\A}(i')+2\varepsilon k \le \left(1+\frac{2\varepsilon}{1-2\varepsilon}\right)cost_{\A}(i')\le \left(1+4\varepsilon \right)cost_{\A}(i') ,$$ 
    where the last inequality follows from the fact that $\varepsilon<1/4$.

\end{proofof}

\begin{proofof}{Claim~\ref{claim:stageII}}
    The proof of Claim~\ref{claim:stageII} involves some subtleties. Let us briefly discuss why before delving into the detailed proof. We aim to replicate the ideas presented in the proof of Claim~\ref{claim:stageI} and identify the point of failure. Consider the $i$-th iteration of the algorithm and assume that $i$ has a partial and not perfect matching. By applying the previous arguments, we deduce that the costs of the blocks prior to $i$ that have a partial match to $j$ do not change significantly. However, after removing these matches, most of the characters in the $j$-th block of $E(Y)$ will be deleted. In the modified alignment $\A_{\mathrm{II}}$, the assigned cost to the $i$-th block of $E(X)$ may account for all these deletions (in both $E(X)$ and $E(Y)$). Therefore, the cost of the $i$-th block may increase from approximately $k$ in $\A_{\mathrm{I}}$ to approximately $2k$ in $\A_{\mathrm{II}}$, resulting in a potential 2-factor increase in the cost. Thus, a more global argument is required, one that doesn't analyze the cost of each block independently.

    Indeed,  let us examine the $i$-th iteration of the algorithm and assume that $i$ has a partial and not perfect match. Let $j$ be the smallest block matched to $i$ under $\A_{\mathrm{I}}$. For simplicity we refer to the $i$-th block $E(X)$ as the $i$-th block and the $j$-th block of $E(Y)$ as the $j$-th block.
    
    Firstly, note that under $\A_{\mathrm{I}}$, if there are characters in the $i$-th block matched outside the $j$-th block, then they must be matched within blocks larger than $j$. The choice of $j$ ensures that the partial matching blocks of $i$ can only be to blocks greater or equal to $j$. Characters in the $j$-th block matched outside the $i$-th block must be matched within blocks larger than $i$, since otherwise, the algorithm should have deleted the matching characters between the $i$-th and $j$-th blocks in previous iterations. 
    
    During the $i$-th iteration, we claim that only one of two possibilities arises: either the $i$-th block contains more than a single partial match, or the $j$-th block contains more than a single partial match.  Simultaneous occurrence of these possibilities is precluded, as such a scenario would violate the monotonicity of the alignment.
    We will proceed to prove these cases separately.
    
    \begin{description}
        \item[Case I:  the $i$-th block contains more than a single partial match:]
        Initially, note that in this scenario, during the $i$-th iteration, the algorithm removes all the matched characters associated with the $i$-th block.
    Let $m$ represent the number of blocks that have a partial match with the $i$-th block of $E(X)$ under $\A_{\mathrm{I}}$. Observe that in $\A_{\mathrm{II}}$ we delete all the matched characters in these $m$ blocks which may contribute an extra factor of at most: $2\varepsilon mk$. 

    Now, if $m=1$, then, under $\A_{\mathrm{I}}$, at least $(1-\varepsilon)k$ characters were deleted in both the $i$-th block of $E(X)$ and the $j$-th block of $E(Y)$. Due to the definition of the $i$-th segment of $E(Y)$ and since the first matched character in the $j$-th block is matched to the $i$-th block, the $i$-th segment of $E(Y)$ contains the $j$-th block, and hence $cost_{\A_{\mathrm{I}}}(i)\le (2-2\varepsilon)k$. 
    Thus, and since $\varepsilon<1/4$: 
    $$cost_{\A_{\mathrm{II}}}(i)\ge cost_{\A_{\mathrm{I}}}(i)-2\varepsilon k \ge cost_{\A_{\mathrm{I}}}(i)\left(1 -\frac{\varepsilon }{1-\varepsilon}\right)\ge  cost_{\A_{\mathrm{I}}}(i)(1-2\varepsilon)$$

    For $m> 1$, under $\A_{\mathrm{I}}$, we observe that: (i) at least $(1-m\varepsilon)k$ characters of the $i$-th block are deleted; (ii) for a block in $E(Y)$ that has a partial match to $i$, except perhaps for the last block, at least $(1-\varepsilon)k$ of its characters are deleted (as per the monotonicity property, its characters are matched only to the $i$-th block). Hence, we have:
    $$cost_{\A_{\mathrm{I}}}(i)\ge (1-m\varepsilon)k+(m-1)(1-\varepsilon)k\ge (1 -2 \varepsilon)m k,  $$
    Here, the first term is due to the deletion in the $i$-th block, and the second term is due to the $E(Y)$ blocks. Combining these, we obtain:
    $$cost_{\A_{\mathrm{II}}}(i)\ge cost_{\A_{\mathrm{I}}}(i)-2\varepsilon mk \ge cost_{\A_{\mathrm{I}}}(i)\left(1 -\frac{2\varepsilon }{1-2\varepsilon}\right)\ge cost_{\A_{\mathrm{I}}}(i)(1-4\varepsilon)$$

    \item[Case II: The $j$-th block contains more than a single partial match:]
    Alternatively, if $i$ has a single partial match to $j$, and $j$ has multiple ones, then the algorithm deletes all the matched characters involving the $j$-th block.
    Here, the argument involves evaluating the cost of all $E(X)$ blocks that are partially matched into the $j$-th block of $E(Y)$. The claim is that their total cost does not change much, and the proof follows a similar approach to that in Case I. In particular if $j$ has a partial matching with the blocks $i_1<\ldots <i_m$, we claim that: $\sum _{k=1\dots m}cost_{\A_{\mathrm{II}}}(i_k)\ge (m-1)(1-\varepsilon)k$. The proof can be concluded using the same idea of the previous case, we omit the details. Hence:
$$\sum _{k=1 \dots}cost_{\A_{\mathrm{II}}}(i_k) \ge\sum _{k=1 \dots m}cost_{\A_{\mathrm{I}}}(i_k)-2\varepsilon mk \ge  (1-4\varepsilon)\cdot \left(\sum _{k=1 \dots m}cost_{\A_{\mathrm{I}}}(i_k)\right). $$

In summary, in both cases, the cost of affected blocks in each iteration does not decrease beyond a $(1-4\varepsilon)$ fraction of their original cost. Due to the disjoint nature of the set of affected blocks over different iterations, we conclude that the total cost of the new alignment is at least $(1-4\varepsilon)$ of the cost of the original alignment, as asserted.
        
    \end{description}
\end{proofof}

\section{Reduction of formula evaluation problem to LCS}\label{appendix:LCS-formula}
\begin{lemma}[OR-gadget]\label{lem:ORgadget}
    Let $k,T,F\ge 1$ be integers where $k$ is even and $k/2 \le F < T$.
    Let $X_0,Y_0,X_1,Y_1 \in \{0,1\}^k$ be balanced strings where $LCS(X_0,Y_0), LCS(X_1,Y_1)\in \{T,F\}$.
    Let
    \begin{alignat}{9}
    X' & {} = {} & {} 1^{k/2} {} & \,\,\, 1^{4k} {} & \,\,\, X_0 {} {} & \,\,\, 1^{4k} {} & \,\,\, 0^{4k} {} & \,\,\, X_1 {} & \,\,\, 0^{4k} {} & \,\,\, 0^{k/2} , \\
    Y' & {} = {} & {} 0^{k/2} {} & \,\,\, 0^{4k} {} & \,\,\, Y_1 {} {} & \,\,\, 0^{4k} {} & \,\,\, 1^{4k} {} & \,\,\, Y_0 {} & \,\,\, 1^{4k} {} & \,\,\, 1^{k/2} .
    \end{alignat}
    Let $T'=9k+T$ and $F'=9k + F$.
    If $LCS(X_0,Y_0)=T$ or $LCS(X_1,Y_1) = T$ then $LCS(X',Y')=T'$, and otherwise $LCS(X',Y')=F'$.
\end{lemma}

\begin{proofof}{Lemma \ref{lem:ORgadget}}
The analysis considers how the maximum matching between $X'$ and $Y'$ can look like.
For ease of notation, we divide $X'$ and $Y'$ into 8 blocks each so $X'=b^{X'}_1b^{X'}_2\dots b^{X'}_8$ and $Y'=b^{Y'}_1b^{Y'}_2\dots b^{Y'}_8$ 
where:
\begin{alignat}{22}
       &         & {} b^{X'}_1     \, & \,\,\, b^{X'}_2    {} & \,\,\, b^{X'}_3 {} & \,\,\, b^{X'}_4    {} & \,\,\, b^{X'}_5    \, & \,\,\, b^{X'}_6 {} & \,\,\, b^{X'}_7    \, & \,\,\, b^{X'}_8  \\
    X' & {} = {} & {} 1^{k/2} {} & \,\,\, 1^{4k} {} & \,\,\, X_0 {} & \,\,\, 1^{4k} {} & \,\,\, 0^{4k} {} & \,\,\, X_1 {} & \,\,\, 0^{4k} {} & \,\,\, 0^{k/2} , \\
    Y' & {} = {} & {} 0^{k/2} {} & \,\,\, 0^{4k} {} & \,\,\, Y_1 {} & \,\,\, 0^{4k} {} & \,\,\, 1^{4k} {} & \,\,\, Y_0 {} & \,\,\, 1^{4k} {} & \,\,\, 1^{k/2} \\
       &         & {} b^{Y'}_1     \, & \,\,\, b^{Y'}_2    {} & \,\,\, b^{Y'}_3 {} & \,\,\, b^{Y'}_4    {} & \,\,\, b^{Y'}_5    \, & \,\,\, b^{Y'}_6 {} & \,\,\, b^{Y'}_7    \, & \,\,\, b^{Y'}_8
\end{alignat}
that is $b^{X'}_3=X_0$, $b^{X'}_6=X_1$, $b^{Y'}_3=Y_1$, $b^{Y'}_6=Y_0$ and all the other $b^{X'}_i$'s and $b^{Y'}_i$'s are either blocks of 0's or 1's.

First we consider two significant matchings of $X'$ and $Y'$.
Consider a matching that matches symbols from $b^{X'}_1$ to symbols in $b^{Y'}_3$, $b^{X'}_2$ to $b^{Y'}_5$, $b^{X'}_3$ to $b^{Y'}_6$, $b^{X'}_4$ to $b^{Y'}_7$ and $b^{X'}_6$ to $b^{Y'}_8$.
Since $LCS(b^{X'}_3,b^{Y'}_6)=LCS(X_0,Y_0)$, $LCS(b^{X'}_1,b^{Y'}_3)=LCS(Y_1,1^{k/2})=k/2$ and $LCS(b^{X'}_6,b^{Y'}_8)=LCS(X_1,1^{k/2})=k/2$ 
largest such matching has size $8k + LCS(X_0,Y_0)+k/2+k/2 = 9k + LCS(X_0,Y_0)$.
Notice, this is either $T'$ or $F'$.
Similarly, a matching that matches symbols from $b^{X'}_3$ to $b^{Y'}_1$, $b^{X'}_5b^{X'}_6b^{X'}_7$ to $b^{Y'}_2b^{Y'}_3b^{Y'}_4$ and $b^{X'}_8$ to $b^{Y'}_6$ will have size $9k + LCS(X_1,Y_1)$.
Thus, we only need to argue that there is no matching larger that $9k+\max(LCS(X_0,Y_0),LCS(X_1,Y_1))$.

Consider a maximum matching of $X'$ to $Y'$.
Assume that the matching starts by matching a symbol 1 in $X'$ and $Y'$.
Without changing the cost of the matching we can replace that edge by matching the left-most 1's in $X'$ and $Y'$.
In $X'$, the first 1 is in $b^{X'}_1$ so we can assume that there is a symbol from $b^{X'}_1$ which is matched somewhere.
Thus the blocks $b^{Y'}_1$ and $b^{Y'}_2$ are unmatched.
Furthermore, one can assume without loss of generality that all symbols matched from $b^{X'}_1$ are matched to 1's in $b^{Y'}_3=Y_0$
as $Y_0$ contains enough 1's and replacing each matched pair between $b^{X'}_0$ and $b^{Y'}_3\cdots b^{Y'}_8$ by a matching pair 
which uses the leftmost unmatched 1 in $b^{Y'}_3$ (going over edges from left to right) will not affect the cost of the matching.
So without loss of generality all matched symbols from $b^{X'}_1$ are matched into $b^{Y'}_3=Y_1$.


We are now concerned with matching symbols after $b^{X'}_1$. Matching any symbol from $b^{Y'}_4$ will block at least $4k$ 1's in $b^{X'}_1b^{X'}_2$ from being matched anywhere. 
If $b^{Y'}_4$ is matched only to $b^{X'}_3$ then we will match at most $2k$ 0's  from $b^{Y'}_3,b^{Y'}_4, b^{Y'}_6$, 
and at most $5k$ 1's from $b^{X'}_3\cdots b^{X'}_8$ so
altogether the matching will be of size at most $7k$. 
If $b^{Y'}_4$ is matched to something in $b^{X'}_5\cdots b^{X'}_8$ then at most $k/2$ 1's from $b^{Y'}_5\cdots b^{Y'}_8$ will be matched
so even if all $5k$ 0's from $b^{Y'}_3\cdots b^{Y'}_8$ were matched 
the overall matching will be of size at most $6k$.
Similar argument applies if we were to match 0's from $b^{Y'}_3$, the matching would be smaller than $7k$.
Thus we can assume that no 0 is matched from $b^{Y'}_3b^{Y'}_4$ if $b^{X'}_1$ matches something.

Hence, we can assume that if the matching starts by 1 then 
our maximum matching matches 1's in $b^{X'}_1b^{X'}_2$ greedily from left to right
that is $b^{X'}_1b^{X'}_2$ is completely matched to 1's in $b^{Y'}_3$ and $b^{Y'}_5$ which gives $4.5k$ 1's.
It remains to match $b^{X'}_3\cdots b^{X'}_8$ to $b^{Y'}_6b^{Y'}_7b^{Y'}_8$. Matching 0's from $b^{Y'}_6=Y_0$ to anywhere in $b^{X'}_5\cdots b^{X'}_8$
will cut-off all but $k/2$ 1's in $b^{Y'}_7b^{Y'}_8$ from being matched.
This would result in a small matching.
Hence, $b^{Y'}_6$ must match $b^{X'}_3b^{X'}_4$ so no 0 in $b^{X'}_5\dots b^{X'}_8$ will be matched.
Thus we can assume that all 1's in $b^{Y'}_7b^{Y'}_8$ are greedily matched from right to left into $b^{X'}_6$ and then $b^{X'}_4$.


This leaves $b^{X'}_3=X_0$ to be matched to $b^{Y'}_6=Y_0$.
This can be done by a matching of size $LCS(X_0,Y_0)$.
This is one of the two good matching we have considered initially so it has size $9k + LCS(X_1,Y_1)$.

If our matching starts by 0, a completely symmetric argument gives that its size is going to be $9k + LCS(X_1,Y_1)$.
So there is no matching larger than $9k + \max\{LCS(X_0,Y_0),LCS(X_1,Y_1)\}$.
It follows that the cost of the largest matching is $F'$ if both  $LCS(X_0,Y_0)=LCS(X_1,Y_1)=F$ and $T'$ otherwise.
\end{proofof}

Notice, both $X'$ and $Y'$ are balanced in the above lemma.
Also, $|X'|=|Y'|=19k$.

\begin{lemma}[AND-gadget]\label{ANDgadget}
    Let $k,T,F\ge 1$ be integers where $k$ is even and $k/2 \le F < T$.
    Let $X_0,Y_0,X_1,Y_1 \in \{0,1\}^k$ be balanced strings where $LCS(X_0,Y_0), LCS(X_1,Y_1)\in \{T,F\}$.
    Let
    \begin{alignat}{11}
    X' & {} = {} & 0^{T+F} & \,1^{11k+T+F} & \, 0^{5k} & \, X_0 & \, 0^k & \, 1^k & \, 0^k & \, X_1 & \, 0^{5k} &,\\
    Y' & {} = {} & 0^{T+F} &   & \, 0^{5k} & \, Y_0 & \, 0^k & \, 1^k & \, 0^k & \, Y_1 & \, 0^{5k} & \, 1^{11k+T+F}.
    \end{alignat}
    Let $T'=13k+3T+F$ and $F'=13k + 2T + 2F$.
    If $LCS(X_0,Y_0) = LCS(X_1,Y_1) = T$ then $LCS(X',Y')=T'$, and otherwise $LCS(X',Y')=F'$.
\end{lemma}

\begin{proofof}{Lemma \ref{ANDgadget}}
Again, the proof proceeds by a case by case analysis of a maximum matching between $X'$ and $Y'$.
WLOG we can focus only on maximum matchings that match the initial block $0^{T+F}$ in $X'$ to the same initial block in $Y'$.
We divide $X'$ and $Y'$ into blocks:
\begin{alignat}{22}
       &         & {} b^{X'}_0     \,\,\,\, & \,\,\, b^{X'}_8   \,\,\, & \,\,\, b^{X'}_1 {} & \,\,\, b^{X'}_2    {} & \,\,\, b^{X'}_3    \, & \,\,\, b^{X'}_4 {} & \,\,\, b^{X'}_5    \, & \,\,\, b^{X'}_6 & \,\,\, b^{X'}_7 \\
    X' & {} = {} & 0^{T+F}       & \,\,\,1^{11k+T+F} & \,\,\, 0^{5k} & \,\,\, X_0 & \,\,\, 0^k & \,\,\, 1^k & \,\,\, 0^k & \,\,\, X_1 & \,\,\, 0^{5k} &,\\
    Y' & {} = {} & 0^{T+F} &   & \,\,\, 0^{5k} & \,\,\, Y_0 & \,\,\, 0^k & \,\,\, 1^k & \,\,\, 0^k & \,\,\, Y_1 & \,\,\, 0^{5k} & \,\,\, 1^{11k+T+F} \\
       &         & {} b^{Y'}_0     \,\,\,\, &                   & \,\,\, b^{Y'}_1 {} & \,\,\, b^{Y'}_2    {} & \,\,\, b^{Y'}_3    \, & \,\,\, b^{Y'}_4 {} & \,\,\, b^{Y'}_5    \, & \,\,\, b^{Y'}_6 \, & \,\,\, b^{Y'}_7 \, & \,\,\, b^{Y'}_8
\end{alignat}

There are two significant matchings between $X'$ and $Y'$ we will analyze first.
Consider a matching that matches all 1's in $X'$ and $Y'$.
Such a matching necessarily matches some 1 from block $b^{X'}_8$ to a 1 in $b^{Y'}_8$.
That prevents all the 0's from $b^{X'}_1\dots b^{X'}_7$ to be matched and similarly for 0's from 
$b^{Y'}_1\dots b^{Y'}_7$.
Hence, the size of such a matching is exactly $T+F+(11k+T+F)+k/2+k+k/2 = 13k+2T+2F=F'$, the number of ones plus the size of $b^{X'}_0$. 
It is thus clear, that any matching that matches some 1 from $b^{X'}_8$ to a 1 in $b^{Y'}_8$ has size at most $F'$.

The other significant matching is a matching that matches $b^{X'}_i$ to $b^{Y'}_i$, for $i=1,\dots,7$, so that each  pair of blocks is matched in the best possible way.
Such a matching has size $T+F+13k+LCS(X_0,Y_0)+LCS(X_1,Y_1)$.
Thus if $LCS(X_0,Y_0) = LCS(X_1,Y_1) = T$ we get an overall matching of size $13k+3T+F=T'$.

Now, our goal is to argue that any maximum matching of $b^{X'}_1\dots b^{X'}_7$ 
to $b^{Y'}_1\dots b^{Y'}_7$ has size at most $13k+LCS(X_0,Y_0)+LCS(X_1,Y_1)$.
Since the length of $b^{X'}_1\dots b^{X'}_7$ and $b^{Y'}_1\dots b^{Y'}_7$ is $15k$
and $13k+LCS(X_0,Y_0)+LCS(X_1,Y_1) \ge 14k$ if a matching leaves at least $k+1$ symbols in either one of the strings unmatched it cannot be maximum.
If a matching would match a 1 from the central block $b^{X'}_4$ to either $b^{Y'}_2=Y_0$ or $b^{Y'}_6=Y_1$, at least $k+1$ symbols would be left unmatched in $b^{Y'}_1\dots b^{Y'}_7$
so the matching would not be maximum.
So a maximum matching must match 1's in the central block $b^{X'}_4$ to 1's 
in $b^{Y'}_4$ if it matches them at all.
Clearly, the best is to match all of them.
(If a matching would not match the central 1's then it can either be increased by matching the 1's or not.
In the latter case it must be matching some symbol from $b^{X'}_4$ to a symbol in $b^{Y'}_1b^{Y'}_2b^{Y'}_6b^{Y'}_7$ (or vice versa for $b^{Y'}_4$) hence leaving unmatched at least $k+1$ symbols in one of the strings.)

So a maximum matching of $b^{X'}_1\dots b^{X'}_7$ to $b^{Y'}_1\dots b^{Y'}_7$ matches all 1's in $b^{X'}_4$ and $b^{Y'}_4$ to each other.
Hence, without loss of generality it matches also the neighboring 0's, so $b^{X'}_3b^{X'}_4b^{X'}_5$ is matched to $b^{Y'}_3b^{Y'}_4b^{Y'}_5$.
WLOG we can also assume that $b^{X'}_1$ perfectly matches $b^{Y'}_1$, and $b^{X'}_7$ perfectly matches $b^{Y'}_7$. 
Hence a maximum matching of $b^{X'}_1\dots b^{X'}_7$ to $b^{Y'}_1\dots b^{Y'}_7$
matches the corresponding $b^{X'}_i$ and $b^{Y'}_i$ for $i=1,3,4,5,7$.
The best we can do on the remaining parts consisting of $A_i$'s and $Y_i$'s is $LCS(X_0,Y_0)+LCS(X_1,Y_1)$.
Hence, a maximum matching of $b^{X'}_1\dots b^{X'}_7$ to $b^{Y'}_1 \dots b^{Y'}_7$ has cost 
$13k+LCS(X_0,Y_0)+LCS(X_1,Y_1)$.

It remains to consider a matching that matches some 1 from $b^{X'}_8$ to $b^{Y'}_1\dots b^{Y'}_7$ but not to $b^{Y'}_8$.
Such a matching necessarily has to leave $b^{Y'}_1$ unmatched which is $5k$ symbols.
If $b^{Y'}_8$ does not match anything in $b^{X'}_1\dots b^{X'}_7$, the best matching we can get has size at most $10k=|b^{Y'}_2\dots b^{Y'}_7|$.
If $b^{Y'}_8$ matches something in $b^{X'}_1\dots b^{X'}_7$, it can contribute at most $2k$ additional edges matching 1's. 
Either way, the matching will fall short of the required $13k$ edges.    
\end{proofof}

Again, $X'$ and $Y'$ in the above lemma are balanced. 
Indeed, the initial block of $T+F$ zeros has the sole purpose of making them balanced.
Also $|X'|=|Y'|=26k+2T+2F \le 30k$.

\begin{proofof}{Lemma \ref{lem:fleeval}}
The claims regarding the length of $g(A,\phi)$ and $h(B,\phi)$ follow easily by induction on the depth of the formula
using the fact that the formula is normalized.
All normalized formulas of the same depth that have top gate $AND$ give strings of the same size,
and similarly all normalized formulas of the same depth  that have top gate $OR$ give strings of the same size.
Indeed, in the base case $d=1$, all the output strings are of length 2.
Then either we apply $AND$ composition on strings of the same length generated for the left and right sub-formulas
or we apply $OR$ composition.
In each case the length of the strings for the sub-formulas are the same so the resulting strings are also of the same size.
In each step the length of the strings multiplies by a factor of at most $30$, so the output strings are of length at most $30^d$.

The claim about $LCS(g(A,\phi),h(B,\phi))$ being either $f(\phi)$ or $t(\phi)$ also follows by induction on the depth of $\phi$.
In the base case, $d=1$ and $\phi$ is a literal, so the claim is obvious from the definition of $g(A,\phi)$ and $h(B,\phi)$ for literals.
For higher depths $d>1$, the claim follows inductively by Lemma \ref{ANDgadget} if $\phi$ has top gate $AND$,
or by Lemma \ref{lem:ORgadget} if $\phi$ has top gate $OR$.

The final claim $f(\phi) < t(\phi)$ follows inductively as well. 
\end{proofof}

\begin{proofof}{Theorem \ref{thm:formulas}}
First, we prove the claim regarding $LCS(X,Y)$.
For a string $X\in \Sigma^n$, let $\overline{X}$ be its (natural) binary encoding using $3n \log n$ bits, and for an integer $k\le n$, let $\overline{k}$ be its encoding in unary using $n$ bits.
There is a non-deterministic Turing machine running in logarithmic space that given inputs $X,Y \in \Sigma^n$ and $k\in \mathbb{N}$ checks whether $LCS(X,Y) \ge k$.
This is because the question whether $LCS(X,Y) \ge k$ or not can be efficiently reduced to a reachability question on a directed graph.
For fixed $n$, this non-deterministic computation can be turned into a normalized boolean formula $\phi_n(U,V,W)$ of depth $d=O(\log^2 n)$ 
which takes as its input the binary encoding of $X,Y$ and $k$ 
such that $\phi_n(\overline{X},\overline{Y},\overline{k})$ is true if and only if $LCS(X,Y) \ge k$, where  $\phi_n(\overline{X},\overline{Y},\overline{k})$ is the evaluation of $\phi_n$ with the assignment $U=\overline{X}, V=\overline{Y}$ and $W=\overline{k}$.
Let $M = n \cdot |g(\overline{X},\phi_n(U,V,\overline{1}))|$, where $X$ is an arbitrary string of length $n$.
Notice, $M$ depends only on the depth of $\phi_n(U,V,W)$ not on the actual assignment of variables.
Set $N=(3n-2)M$.

Define 
\begin{eqnarray*}
G(X)= g(\overline X,\phi_n(U,V,\overline{1})) \cdot 0^M1^M \cdot g(\overline X,\phi_n(U,V,\overline{2})) \cdot 0^M1^M \cdots g(\overline X,\phi_n(U,V,\overline{n})), \\
H(Y)= h(\overline Y,\phi_n(U,V,\overline{1})) \cdot 0^M1^M \cdot h(\overline Y,\phi_n(U,V,\overline{2})) \cdot 0^M1^M \cdots h(\overline Y,\phi_n(U,V,\overline{n})).
\end{eqnarray*}

Note that for a fixed value of $k$, $\phi_n(U,V,\overline{k})$ represents a formula that has two sets of variables $U$ and $V$, where $U$ depends only on $X$ and $V$ depends only on $Y$.  Clearly, $G(X)$ depends solely on the string $X$, while $H(Y)$ depends on $Y$.

Observe that for any $k$ smaller or equal than $ LCS(X,Y)$, we have $\phi_n(\overline{X}, \overline{Y},\overline{k})$ is true and hence by Lemma~\ref{lem:fleeval} we get: \[LCS(g(\bar X, \phi_n(U,V,\overline{k} )), h(\bar Y, \phi_n(U,V,\overline {k})))=T\] For $k$ that exceeds $ LCS(X,Y)$ this evaluates to $F$.

Furthermore, we can get a matching between $G(X)$ and $H(Y)$ of size \[2M(n-1)+LCS(X,Y) \cdot T + (n-LCS(X,Y)) \cdot F = 2M(n-1) + nF + LCS(X,Y) \cdot (T-F),\] by matching optimally the consecutive blocks of $G(X)$ and $H(Y)$. 
Any other matching that would try to match $g(\overline X,\phi_n(U,V,\overline{i}))$ to $h(\overline Y,\phi_n(U,V,\overline{j}))$
for $i\neq j$ will leave at least $2M$ symbols unmatched so it will be worse.
Hence, $LCS(G(X),H(Y))=2M(n-1) + nF + LCS(X,Y) \cdot (T-F)$.
Setting $R=2M(n-1) + nF$ and $S=(T-F)$ gives the required relationship.

The proof for $\ED$ is the same.
If we make sure that the normalized formula for $LCS$ and $\ED$ are both of the same depth, then the parameters $S$ and $R$ will be identical.\end{proofof}

\section{Embedding Indel distance metric into Edit distance metric}\label{sec:IndelToEdit}
Tiskin~\cite{Tiskin08} (in section 6.1) first observed the existence of an embedding that maps strings from $\Sigma^n$ to strings in $\left(\Sigma \cup \{\$\}\right)^{2n}$. This embedding satisfies the property that for any $X,Y \in \Sigma^n$, we have $2 \ED(X,Y) = \LCS(E(X),E(Y))$. The embedding is straightforward: it involves appending a special character ``\$'' after every character in its input. 

Let us delve into the intuition behind analyzing distance preservation.
Given any optimal $\ED$-alignment $\A$ of $X$ and $Y$, we can transform it into an $\LCS$-alignment $\A'$ of $E(X)$ and $E(Y)$ effectively doubling its cost, as follows: 
if a character is deleted from either $X$ or $Y$ in $\A$, then the corresponding character along with the following $\$$ in $E(X)$ or $E(Y)$ are deleted in $\A'$.
If $X[i]$ is substituted with $Y[j]$ in $\A$, then the corresponding characters: $E(X)[2i-1]$ and $E(Y)[2j-1]$ are deleted in $\A'$. 
Since each deletion or substitution in $\A$ corresponds to deleting two characters in $\A'$, the cost of $\A'$ is twice the cost of $\A$, i.e., $2\ED(X,Y)$.
Although one must be careful, it is not difficult to prove that the resulted alignment $\A'$ described above is optimal.

Inspired by the embedding outlined earlier, which transforms the $\ED$ metric into the $\LCS$ metric, this section introduces two separate embedding mappings that function conversely: from the $\LCS$ metric to the $\ED$ metric.

In this section we introduce an embedding from the $\LCS$ metric to the $\ED$ metric which is scaling isometric. Our embedding is asymmetric, implying that we embed each string in a different manner. We propose an embedding scheme that transforms strings $X$ and $Y$ of length $n$, into $E_1(X)$ and $E_2(Y)$, respectively of lengths $n$ and $O(n^2)$. This scheme ensures that any optimal $\LCS$-alignment of $X$ and $Y$ corresponds to an optimal $\ED$-alignment of $E_1(X)$ and $E_2(Y)$. The formal statement of this result is provided below.



\begin{theorem}\label{thm:indel-edit-exact}

    For any alphabet $\Sigma$ and integer $n > 0$, there exist $E_1:\Sigma^n\to \Sigma^n$ and $E_2:\Sigma^n\to \{\Sigma\cup \{\$\}\}^{N}$, where $N = \mathcal{O}(n^2)$, such that given strings $X,Y \in \Sigma^n$, we have $\ED(E_1(X),E_2(Y)) = N - n + \frac{\LCS(X,Y)}{2}$. \footnote{Our techniques can adapted easily to design embedding functions for variable length strings instead of length-preserving ones.}
\end{theorem}



The core concept of our embedding revolves around defining $E_1$ as the identity function, while for $E_2(Y)$, appending the sequence $\$^n$ after each character in $Y$, including at the beginning, resulting in a length of $O(n^2)$ for $E_2(Y)$. This construction ensures that any optimal $\LCS$ alignment of $X$ and $Y$ can be transformed into an $\ED$ alignment of $E_1(X)$ and $E_2(Y)$ while preserving matching characters, as elaborated below:

Given any optimal $\LCS$-alignment $\A$ of $X$ and $Y$, we can transform it into an $\ED$-alignment $\A'$ of $E(X)$ and $E(Y)$  as follows: 
If a character is deleted from  $Y$ in $\A$, then the corresponding character along with the subsequent sequence of  $\$^n$ in $E(Y)$ are deleted in $\A'$. If a character is deleted from $X$, then since the each of the characters in $E(Y)$ is separated with the sequence $\$^n$ we can substitute the corresponding character in $E(X)$,  with a $\$$-symbol in $E(Y)$. This ensures that the characters of $E_1(X)$ are either matched or substituted. Consequently, we establish the optimality of this resulting $\ED$ alignment of $E_1(X)$ and $E_2(Y)$. The formal proof with all the details is in Appendix \ref{sec:indel-edit-exact}.



In the previous theorem, the length of one of the embedded strings grows quadratically with the input string's length due to appending $\$^n$ after each character in $Y$ to form $E_2(Y)$. Now, a natural question arises: Can we reduce the length of the embedded string? While we currently lack knowledge of any embedding with a smaller output size which is scaling isometric, we can achieve significantly smaller embedded strings by approximately preserving the distances.

Essentially, instead of appending $\$^n$, we can append $\$^k$ after each character in $Y$ for $k \le n$, resulting in a much smaller-sized embedding while maintaining distances up to a factor of $(1+c/k)$, for some small constant $c \ge 1$. This is due to the claim that even with the smaller appending, one can still achieve a $\ED$ alignment of the embedded strings given an optimal $\LCS$ alignment of the input strings, which preserves more than $(1 - \frac{c}{k})$ fraction of the matches.  We formally state this approximate embedding below, with further details provided in Appendix \ref{sec:indel-edit-apx}. 



\begin{thm}[Indel Into Edit Metrics Embedding - Approximate embedding]
\label{thm:indel-edit-apx}
     For any alphabet $\Sigma$, $n \in \N $ and $ \varepsilon \in (0,1]$, there exist mappings $E_1:\Sigma^n\to \Sigma^n$ and $E_3:\Sigma^n\to (\Sigma\cup \{\$\})^{N}$, where $N = \Theta(n/\varepsilon)$, such that for any $X,Y \in \Sigma^n$, we have
     $$\ED(E_1(X),E_3(Y)) = N - n+k, \text{ where } k\in \left[\frac{\LCS(X,Y)}{2},(1 + \varepsilon)\frac{\LCS(X,Y)}{2}\right).$$
     

\end{thm}




\subsection{Obtaining Scaling Isometric Embedding}\label{sec:indel-edit-exact}
In this section, we describe the embedding functions $E_1$ and $E_2$ for the scaling isometric embedding from $\LCS$ metric to $\ED$ metric, followed by the proof of Theorem \ref{thm:indel-edit-exact}.

The following facts, which we state without a proof, will be helpful in the proof of Theorem \ref{thm:indel-edit-exact} and \ref{thm:indel-edit-apx}.
\begin{fact}\label{fact:indel-edit1}
    For any $\LCS$ alignment of strings $X$ and $Y$, where $\abs{X} = \abs{Y}$, the number of deletions in $X$ is equal to the number of deletions in $Y$ which is equal to $\frac{\LCS(X,Y)}{2}$.
\end{fact}

\begin{fact}\label{fact:indel-edit2}
    For any optimal $\ED$ alignment of strings $X$ and $Y$, where $\abs{X} \le \abs{Y}$, we have $\ED(X,Y) = $ \#deletions in $X$ + \#deletions in $Y$ + \#substitutions $= (\abs{Y} - \abs{X}) + 2\times$\#deletions in $X$ + \#substitutions.
\end{fact}

\begin{fact}\label{fact:indel-edit3}
    Let $\Sigma , \Sigma'$ be alphabets, where $\Sigma \subseteq \Sigma'$ and strings $X,Y \in \Sigma^n$ , $Y' \in {\Sigma'}^{N}$, where $N\ge n$.If $Y'$ is obtained from $Y$ by inserting characters from $\Sigma' \setminus \Sigma$ , then $\ED(X,Y') \ge N-n + \frac{\LCS(X,Y)}{2}$.
\end{fact}

\subsubsection{Description of the embedding functions $E_1$ and $E_2$}\label{sec:indel-edit-exact1}
We define the embedding function $E_2$ for strings of length $n$. 
Given a string $Y \in \Sigma^n$, $E_2(Y)$ is obtained by inserting the sequence $\$^n$ after each symbol in $Y$ and also at the beginning. This yields in $E_2(Y) = \$^n \cdot Y[1] \cdot \$^n \cdot Y[2] \cdot \$^n \ldots Y[n] \$^n$, where $\$^n$ denotes a sequence of $n$ dollar signs ``\$". Essentially, $E_2(Y)[i(n+1)] = Y[i]$ for all $1 \leq i \leq n$, and the remaining positions in $E_2(Y)$ are filled with ``\$". The length of the transformed string $E_2(Y)$ is $N' = n(n+1) + n = n^2 + 2n$.


Regarding $E_1$, it functions as the identity function. Therefore, $E_1(X)$, simply remains the same as $X$, thus $|E_1(X)| = |X| = n$.  


\subsubsection{Proof of Theorem \ref{thm:indel-edit-exact}}\label{sec:indel-edit-exact2}
Given strings $X,Y \in \Sigma^n$, henceforth we will denote $E_1(X)$ simply as $X$ and $E_2(Y)$ as $Y'$. 
To prove Theorem \ref{thm:indel-edit-exact}, we demonstrate that we can construct a $\ED$ alignment of $X$ and $Y'$ with a cost of $N'-n + \frac{\LCS(X,Y)}{2}$ given any optimal $\LCS$ alignment of $X$ and $Y$ with a cost of $\frac{\LCS(X,Y)}{2}$. 

Subsequently, we establish that the constructed $\ED$ alignment is indeed optimal 
. This process involves defining a decomposition for $X$ and $Y$ based on their $\LCS$ alignment, which will then be utilized to construct the required $\ED$ alignment for $X$ and $Y'$.

\textbf{Decomposition of $X$ and $Y$:}
We start by defining blocks for $X$ and $Y$ given a $\LCS$ alignment $\A$ of $X$ and $Y$. 
We partition $X$ into blocks consisting of contiguous substrings. Each block, with the exception of the first one, starts with a contiguous matching segment followed by a contiguous deletion segment. Some blocks may have an empty deletion segment. The initial block does not contain any matching segment. These matching and deletion segments determined based on the $\LCS$ alignment $\A$. The decomposition of $Y$ follows a similar procedure. Left hand side of Figure \ref{fig:indel-edit-exact} and Figure \ref{fig:indel-edit-apx} shows the decomposition according to some optimal $\LCS$ alignment. We formally articulate this observation below.

\begin{observation}\label{obs:indel-edit-exact}
For any $\LCS$ alignment $\A$ of $X,Y \in \Sigma^n$, we can partition $X$ and $Y$ into disjoint blocks based on $\A$, such that:
\begin{enumerate}

\item The number of blocks in $X$ and $Y$ are equal. Hence, $X = b^X_0 \cdot b^X_1 \cdot b^X_2 \cdots b^X_l$ and $Y = b^Y_0 \cdot b^Y_1 \cdot b^Y_2 \cdots b^Y_l$, where $b^X_i$ are blocks of $X$ and $b^Y_i$ are blocks of $Y$.

\item $b^X_0 = d^X_0$ and $b^Y_0=d^Y_0$, which may be empty.

\item For each $i > 0$, block $b^X_i$ consists of a contiguous matching part $m^X_i$ followed by a contiguous deletion part $d^X_i$ i.e. $b^X_i = m^X_i \cdot d^X_i$. Similarly, for each block $b^Y_i$ of $Y$ we have, $b^Y_i = m^Y_i \cdot d^Y_i$. This means that in the alignment $\A$, the characters in $m^X_i$ and $m^Y_i$ are getting matched and characters in $d^X_i$ and $d^Y_i$ are getting deleted.

\item For each $i>0$, $m^X_i = m^Y_i$.

\item For each $i > 0$, $m^X_i$ is non-empty and $d^X_i$ can be empty. The same applies for the blocks of $Y$.

\end{enumerate}
\end{observation} 

To define the necessary $\ED$ alignment between $X$ and $Y'$, we begin by establishing a decomposition for both $X$ and $Y'$. This decomposition is derived from an optimal $\LCS$ alignment of $X$ and $Y$. Let $\I^{X,Y}$ be an optimal $\LCS$ alignment of $X$ and $Y$.



\textbf{Decomposition of $X$ and $Y'$:}
 According to Observation \ref{obs:indel-edit-exact}, we have a decomposition of $X$ and $Y$ w.r.t $\I^{X,Y}$. Let $X = b^X_0 \cdot b^X_1 \cdot b^X_2 \cdots b^X_l$ and $Y = b^Y_0 \cdot b^Y_1 \cdot b^Y_2 \cdots b^Y_l$. 


Decomposition of $X$ remains the same. We define the blocks of $Y'$ as $b^{Y'}_1, b^{Y'}_2, \dots , b^{Y'}_l$, where $b^{Y'}_i = m^{Y'}_i \cdot c_{i}' \cdot d^{Y'}_i$, with $c_{i}' = \$^n$. If $m^Y_i = Y[p,q]$, then $m^{Y'}_i = Y'[p(n+1),q(n+1)]$ which is simply  $Y[p] \cdot \$^n \cdot Y[p+1] \cdot \$^n \cdots Y[q]$. Similarly, if $d^Y_i = Y[p,q]$, then $d^{Y'}_i = Y'[p(n+1),q(n+1)+n]$ which is $Y[p] \cdot \$^n \cdot Y[p+1] \$^n \cdots Y[q] \cdot \$^n$ (note the extra $\$^n$ at the end of $d^{Y'}_i$ which is absent in $m^{Y'}_i$). Therefore, $Y' = \$^n \cdot b^{Y'}_0 \cdot b^{Y'}_1 \cdot b^{Y'}_2 \cdots b^{Y'}_l$. Reader can refer to Figure \ref{fig:indel-edit-exact} for more clarity.

Based on this decomposition of $X$ and $Y'$, we are now ready to define our $\ED$ alignment of $X$ and $Y'$.

\textbf{$\ED$ alignment of $X$ and $Y'$:}
We proceed by defining a $\ED$ alignment $\A'$ of $X$ and $Y'$ using the previously described blocks from $X$ and $Y'$.
Initially, we align the block $b^X_0$ with the initial $\$^n$ block, where each character of $b^X_0$ undergo substitution.
Subsequently, for each block $b^X_i$ where $i>0$, we align $m^{X}_i$ with $m^{Y'}_i$, ensuring that all characters in $m^{X}_i$ are matched. Moreover, each $d^{X}_i$ is aligned with $c_{i}'$, resulting in substitutions for each character in $d^{X}_i$. For a visual representation, please refer to Figure \ref{fig:indel-edit-exact}. Here, we observe that $d^X_0$ aligns with the initial $\$^n$ block, each $m^X_i$ aligns with $m^{Y'}_i$, and $d^X_2$ and $d^X_3$ align with $c_2$ and $c_3$ in $Y'$. Finally, we compute the cost of the alignment $\A'$.


\begin{claim}\label{claim:exact1}
    $cost(\A') = N-n+\frac{\LCS(X,Y)}{2}$.
\end{claim}

\begin{proof}
    Since, $|Y'| > |X|$, therefore there must be at least $|Y'| - |X| = N-n$ many deletions in any alignment of $X$ and $Y'$.

    In the optimal $\LCS$ alignment $\I^{X,Y}$ of $X$ and $Y$, we have $m^X_i = m^Y_i$, as noted in Observation \ref{obs:indel-edit-exact}. Since $m^Y_i$ is a proper subsequence of $m^{Y'}_i$ from our embedding and decomposition, it follows that $m^{X}_i$ is also a proper subsequence of $m^{Y'}_i$. Therefore, all characters in $m^{X}_i$ match when aligned with $m^{Y'}_i$. Furthermore, the characters of $d^{X}_i$ align with ``\$" characters in $c_{i}'$, contributing to $|d^{X}_i|$ substitutions. Hence, the total number of substitutions is given by $\sum_{i=0}^l |d^{X}_i| = \frac{\LCS(X,Y)}{2}$, using Fact \ref{fact:indel-edit1}. Since, in the $\ED$ alignment, there are only matches and substitutions in $X$ and all deletions occur only in $Y'$, therefore, $cost(\A') = N-n + \frac{\LCS(X,Y)}{2}$.
\end{proof}

Let $\E^{X,Y'}$ be an optimal $\ED$ alignment of $X$ and $Y'$ and the cost of the alignment is $cost(\E^{X,Y'})$, i.e., $cost(\E^{X,Y'}) = \ED(X,Y')$. 
We will now demonstrate that $\A'$ is an optimal $\ED$ alignment, i.e. $cost(\A') = cost(\E^{X,Y'})$.

\begin{claim}\label{claim:exact2} 
    $cost(\E^{X,Y'}) = cost({\A'})$.
\end{claim}
\begin{proof}
    It is clear that $cost(\E^{X,Y'}) \le cost({\A'})$, because $\E^{X,Y'}$ is an optimal $\ED$ alignment and $\A'$ is an $\ED$ alignment.

    Let us assume, for the sake of contradiction, that $cost(\E^{X,Y'}) < cost(\A')$. In the alignment $\A'$, the characters of $X$ are either matched or substituted. Therefore, the only way to achieve a better alignment than $\A'$ is by increasing the number of matches in $X$. However, matches in $X$ and $Y'$ can only occur with characters that are not ``\$". This implies that such matches would also exist between $X$ and $Y$, which contradicts the optimality of the 
    $\LCS$ alignment of $X$ and $Y$. Hence, we establish that $cost(\E^{X,Y'}) \ge cost(\A')$. Therefore, $cost(\E^{X,Y'}) = cost(\A')$.
\end{proof}

\begin{proofof}{Theorem \ref{thm:indel-edit-exact}}
    Using claim \ref{claim:exact1} and claim \ref{claim:exact2}, we can prove the theorem. 
\end{proofof}

\begin{remark}
    If there exists an $\LCS$ alignment of $X$ and $Y$, where the size of each deletion segment of $X$, i.e., $|d^X_i|$, for all $1\le i\le l$ is bounded by some threshold $t$, then we can get an embedding of size $\OO(nt)$.
\end{remark}

\begin{figure}[htp]
    \centering
    \includegraphics[width=\textwidth]{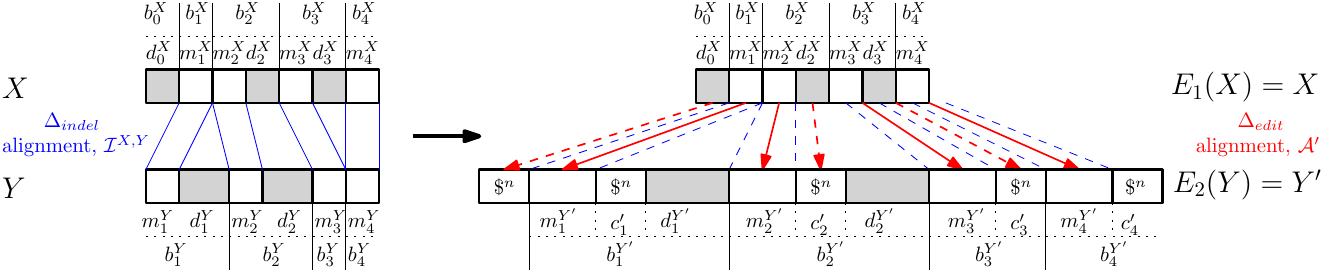}
    \caption{On the left side, we have the decomposition of $X$ and $Y$ based on the $\LCS$ alignment. On the right side, we see the decomposition and alignment of $X$ and $Y'$ following our construction in Section \ref{sec:indel-edit-exact}. The solid red arrows indicate that all characters of $m^{X}_i$ are matched, while the dotted red arrows suggest that the characters of $m^{X}_i$ are substituted. The shaded cells in gray indicate deletions. On the right-hand side, the blue dotted lines indicate the alignment of $m^{X}_i$ with $m^{Y'}_i$.}
    \label{fig:indel-edit-exact}
\end{figure}

\subsection{Obtaining Approximate Scaling Isometric Embedding}\label{sec:indel-edit-apx}


In this section, we introduce the embedding function $E_3$ and provide a proof for Theorem \ref{thm:indel-edit-apx}. While the embedding described here bears resemblance to the previous one, for the sake of completeness, we present all the details below:

\subsubsection{Description of the function $E_3$}\label{sec:indel-edit-apx1}
Given a string $Y$ of length $n$ in the alphabet $\Sigma$, we define $E_3$ similar to $E_2$ as described in Section \ref{sec:indel-edit-exact1}. Given a parameter $0 < \varepsilon \le 1$, we define $k = \frac{4}{\varepsilon}$.\footnote{It's worth noting that $k$ must be an integer. There are two approaches to ensure this: either we choose $\varepsilon$ such that $k$ is an integer, or we use $k \lceil \frac{4}{\varepsilon} \rceil$. However, in the latter case, additional attention is required for the calculations, although the method remains valid.} 
However, for the mapping $E_3$, we append $\$^k$ after every character in $Y$ instead of appending $\$^n$. Additionally, we append $\$^n$ at the beginning and at the end as well. Consequently, $E_3(Y)$ is represented as: $$E_3(Y) = \$^n \cdot Y[1] \cdot \$^k \cdot Y[2] \cdot \$^k \cdot Y[3]  \cdot \$^k \cdots \$^k \cdot Y[n] \cdot \$^k \cdot \$^n. $$ Basically, $E_3(Y)[n + (i-1)(k+1) + 1] = Y[i]$ for all $1 \le i \le n$, and rest of the positions in $E_3(Y)$ are filled with ``\$". The length of $E_3(Y)$ is calculated as $n(k+1) + 2n = \widetilde{N}$. Therefore, for a fixed value of \(\varepsilon\), we achieve a linear size embedding instead of the quadratic embedding previously.

\subsubsection{Proof of Theorem \ref{thm:indel-edit-apx}}\label{sec:indel-edit-apx2}
Given $X,Y \in \Sigma^n$, from now on we denote $E_3(Y)$ by $\widetilde{Y}$ and as in the previous embedding in Section \ref{sec:indel-edit-exact}, we denote $E_1(X)$ by $X$.
The proof proceeds as follows: Let us consider an optimal $\LCS$ alignment $\I^{X,Y}$ of $X$ and $Y$. We aim to construct a $\ED$ alignment $\widetilde{\A}$ of strings $X$ and $\widetilde{Y}$ based on the $\LCS$ alignment $\I^{X,Y}$, which preserves ``most" of the matches. Subsequently, we will bound the $cost(\widetilde{\A})$ to show that we can approximately preserve the distances, as stated in Theorem \ref{thm:indel-edit-apx}. 

Similar to the proof of Theorem \ref{thm:indel-edit-exact}, here also we will utilize the decomposition of $X$ and $Y$ given the optimal $\LCS$ alignment $\I^{X,Y}$. We will employ the decomposition from Section \ref{sec:indel-edit-exact2}. Let $X = b^X_0 \cdot b^X_1 \cdot b^X_2 \cdots b^X_l$ and $Y = b^Y_0 \cdot b^Y_1 \cdot b^Y_2 \cdots b^Y_l$. We will decompose $X$ and $\widetilde{Y}$ based on the decomposition of $X$ and $Y$. The resulting decomposition of $X$ and $\widetilde{Y}$ will then be used to establish the $\ED$ alignment $\widetilde{\A}$.

\textbf{Decomposition of $X$ and $\widetilde{Y}$:}
Now, let us define the blocks of $X$ and $\widetilde{Y}$ based on the obtained blocks of $X$ and $Y$. Blocks of $X$ remain the same.

Let blocks of $\widetilde{Y}$ be $b^{\widetilde{Y}}_0, b^{\widetilde{Y}}_1, b^{\widetilde{Y}}_2, \dots , b^{\widetilde{Y}}_l$, where $b^{\widetilde{Y}}_0 = d^{\widetilde{Y}}_0$ and for all $i>0, b^{\widetilde{Y}}_i = m^{\widetilde{Y}}_i \cdot \widetilde{c}_i \cdot d^{\widetilde{Y}}_i$ , $\widetilde{c}_i = \$^k$. If $m^Y_i = Y[p,q]$, then $m^{\widetilde{Y}}_i = \widetilde{Y}[n + (p-1)(k+1) + 1 , n + (q-1)(k+1) + 1]$, which is $Y[p] \cdot \$^k \cdot Y[p+1] \$^k \cdots \$^k \cdot Y[q]$. If $d^Y_i = Y[p,q]$, then $d^{\widetilde{Y}}_i = \widetilde{Y}[n + (p-1)(k+1) + 1 , n + (q-1)(k+1) + 1 + k]$, which is $Y[p] \cdot \$^k \cdot Y[p+1] \$^k \cdots \$^k \cdot Y[q] \cdot \$^k$ (note the extra $\$^k$ at the end of $d^{\widetilde{Y}}_i$ which is missing in $m^{\widetilde{Y}}_i$). Therefore, $\widetilde{Y} = \$^n \cdot b^{\widetilde{Y}}_0 \cdot b^{\widetilde{Y}}_1 \cdot b^{\widetilde{Y}}_2 \cdots b^{\widetilde{Y}}_l \cdot \$^n$. Refer to Figure \ref{fig:indel-edit-apx} for an example.

\textbf{$\ED$ alignment of $X$ and $\widetilde{Y}$:}
Given the decomposition of $X$ and $\widetilde{Y}$, we can now outline our approach for constructing the $\ED$ alignment $\widetilde{\A}$. The fundamental principle of our alignment strategy is to sequentially align the characters of $X$ with those of $\widetilde{Y}$ from left to right, while carefully accounting for substitutions and deletions.
 
We initiate the alignment process by aligning $b^{X}_0$ with the initial $\$^n$ in $\widetilde{Y}$. Since the length of $b^{X}_0$ is at most $n$, and none of its characters are ``\$", we perform substitutions for all characters in $m^{X}_0$.

For each subsequent block $b^{X}_i = m^{X}_i \cdot d^{X}_i$, from left to right, starting from block 1, we first attempt to align $m^{X}_i$ with $m^{\widetilde{Y}}_i$, matching as many characters of $m^X_i$ as possible and deleting any unmatched characters in $m^{X}_i$.


Following this, we align the characters of $d^{X}_i$ from left to right with the unaligned characters of $\widetilde{Y}$, beginning from the leftmost unaligned character in $\widetilde{Y}$. All the steps are detailed in Algorithm \ref{alg:indel-edit-apx}. 

\begin{algorithm}[H]
\caption{Getting an alignment of $X$ and $\widetilde{Y}$ wrt edit distance metric: }\label{alg:indel-edit-apx}
\KwData{$X = m^{X}_0 \cdot b^{X}_1 \cdot \cdot \cdot b^{X}_l$, where $b^{X}_i = m^{X}_i \cdot d^{X}_i$ and $\widetilde{Y} = \$^n \cdot d^{\widetilde{Y}}_0 \cdot b^{\widetilde{Y}}_1 \cdot \cdot \cdot b^{\widetilde{Y}}_l \cdot \$^n$, where $b^{\widetilde{Y}}_i = m^{\widetilde{Y}}_i \cdot \widetilde{c}_i \cdot d^{\widetilde{Y}}_i$}
\KwResult{An alignment of $X$ and $\widetilde{Y}$}
Align $m^{X}_0$ to the initial $\$^n$;

$j = 1$;

\For{$i=1 \cdots l$}{

  \If{$i \ge j$}{
    \If{$i > j$}{
        \Comment{$m^{\widetilde{Y}}_i$ is fully available.}
        Align $m^{X}_i$ to $m^{\widetilde{Y}}_i$;}
    \If{$i = j$}{\Comment{$m^{\widetilde{Y}}_i$ might be partially available or not available.}
        
        Align $m^{X}_i$ to $m^{\widetilde{Y}}_i$ matching the maximum possible number of characters of $m^{X}_i$;
        
        Delete the characters in $m^{X}_i$ which cannot be matched;
    }

    Align the characters of $d^{X}_i$ from left to right to the unaligned characters of $\widetilde{Y}$, starting from $\widetilde{c}_i$ in $\widetilde{Y}$;
    }

    \If{$i < j$}{

        \Comment{$m^{\widetilde{Y}}_i$ is not available.}
        
        Delete all the characters in $m^{X}_i$.

        Align the characters of $d^{X}_i$ from left to right to the unaligned characters of $\widetilde{Y}$, starting from leftmost unaligned character in $\widetilde{Y}$;
    }

    Update $j$, such that the leftmost unaligned character of $\widetilde{Y}$ is in block $b^{\widetilde{Y}}_j$;
}

\end{algorithm}


Let's initiate the analysis of the $\ED$ alignment $\widetilde{\A}$. First, we examine the alignment of $m^X_i$, and then we proceed to analyze the alignment of $d^X_i$ for all $1 \leq i \leq l$.

While attempting to align $m^{X}_i$ to $m^{\widetilde{Y}}_i$, we encounter three possible scenarios regarding the alignment status of $m^{\widetilde{Y}}_i$. The three scenarios are named as \textbf{fully available}, \textbf{partially available}, and \textbf{not available}. Depending on these three cases we determine how to align the characters of $m^{X}_i$: 
\begin{itemize}
    \item \textbf{fully available}: $m^{\widetilde{Y}}_i$ falls into this category if all its characters are available for alignment. In other words, none of the characters from any $m^{X}_j$, where $j < i$, have been aligned with characters from $m^{\widetilde{Y}}_i$. In this scenario, we align $m^{X}_i$ to $m^{\widetilde{Y}}_i$, matching each character of $m^{X}_i$. This alignment is feasible because $m^X_i = m^Y_i$, and $m^Y_i$ is a proper subsequence of $m^{\widetilde{Y}}_i$ from our construction and the decomposition of $X$ and $\widetilde{Y}$. Thus $m^{X}_i$ is also a proper subsequence of $m^{\widetilde{Y}}_i$. As depicted in Figure \ref{fig:indel-edit-apx}, on the right-hand side, $m^{\widetilde{Y}}_1$ and $m^{\widetilde{Y}}_5$ are fully available, enabling us to match all the characters of $m^{X}_1$ and $m^{X}_5$. 
    
    \item \textbf{partially available}: If some, but not all, characters from $m^{\widetilde{Y}}_i$ are available for alignment, we consider it partially available. This occurs when certain characters from some $d^{X}_j$, where $j < i$, have been aligned with characters from $m^{\widetilde{Y}}_i$. In this case, we align $m^{X}_i$ to $m^{\widetilde{Y}}_i$ and match as many characters of $m^{X}_i$ as possible, deleting the rest from $m^{X}_i$. This is same as matching the largest suffix of $m^{X}_i$ which appears as a proper subsequence of the remaining unaligned suffix of $m^{\widetilde{Y}}_i$. The size of this largest suffix of $m^X_i$ is exactly the number of non-$\$$ characters that are unaligned in $m^{\widetilde{Y}}_i$.

    More technically, let $m^{\widetilde{Y}}_i = \widetilde{Y}[p,q]$. Since, $m^{\widetilde{Y}}_i$ is partially available, therefore, the leftmost unaligned character in $\widetilde{Y}$ is in $m^{\widetilde{Y}}_i$. Let $\widetilde{Y}[p']$ be the leftmost unaligned character in $\widetilde{Y}$. So, $\widetilde{Y}[p',q]$ is available for alignment. Now, we need to find the largest suffix of $m^X_i$ that can match into $\widetilde{Y}[p',q]$. For that let's count the number of non-$\$$ characters in $\widetilde{Y}[p',q]$, which is $\lfloor \frac{(q-p')}{k+1} \rfloor + 1$. The extra 1 non-$\$$ character is $\widetilde{Y}[q]$ because the last character of every $m^{\widetilde{Y}}_j$ is non-$\$$. Now, we match the last $\lfloor \frac{(q-p')}{k+1} \rfloor + 1$ many characters of $m^X_i$ and delete the rest.
    
    In Figure \ref{fig:indel-edit-apx}, $m^{\widetilde{Y}}_2$ is partially unaligned, resulting in two missed matches in $m^{X}_2$, with only the last character of $m^{X}_2$ being matched.
    
    \item \textbf{not available}: $m^{\widetilde{Y}}_i$ is not available if none of its characters are available for alignment. In this case, all characters from some $m^{X}_j$, where $j < i$, have been aligned with all characters from $m^{\widetilde{Y}}_i$. Here, we cannot match any character of $m^{X}_i$, therefore, we delete the entire $m^{X}_i$. We again refer to Figure \ref{fig:indel-edit-apx}, where $m^{\widetilde{Y}}_3$ and $m^{\widetilde{Y}}_4$ are not available, resulting in the failure to match any characters from $m^{X}_3$ and $m^{X}_4$.
\end{itemize}

Let us now examine the alignment of $d^X_i$. As stated earlier , the characters of $d^{X}_i$ are aligned from left to right with the unaligned characters 
of $\widetilde{Y}$, beginning from the leftmost unaligned character in $\widetilde{Y}$. The characters in $d^{X}_i$ either get matched or substituted but are not deleted. During this process, it's possible that we align characters of $d^{X}_i$ with characters in $m^{\widetilde{Y}}_j$ for some $j > i$. Consequently, certain characters in $m^{\widetilde{Y}}_j$, which are not ``\$", become unavailable for matching with the corresponding characters in $m^{X}_j$, resulting in missed matches. These missed matching characters in $X$ may be deleted from $X$ during our alignment. Our goal is to quantify the number of deletions in $X$, which is exactly equal to the number of missed matches caused by aligning the characters of $d^{X}_i$ for all $i>0$ to the matching segments of $\widetilde{Y}$, i.e., $m^{\widetilde{Y}}_j$ for $j>i$. For each $1 \leq i \leq l$, let $S_i$ denote the total number of characters from $d^{X}_i$ that gets aligned with characters from $m^{\widetilde{Y}}_j$ for all $j > i$.

\begin{figure}[htp]
    \centering
    \includegraphics[width=\textwidth]{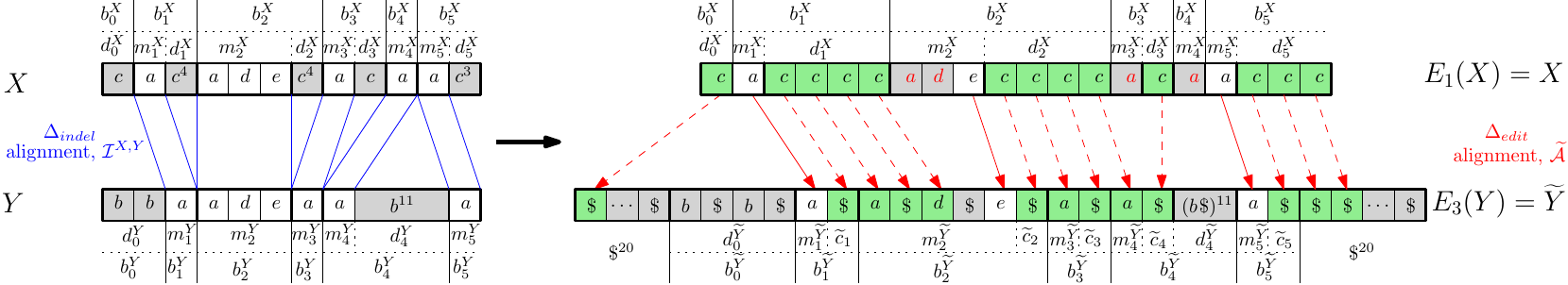}
    \caption{On the left side, we have the decomposition of two strings $X$ and $Y$ of length 20 each based on the $\LCS$ alignment. On the right side, we see the decomposition and alignment of $E_1(X)=X$ and $E_3(Y)=\widetilde{Y}$ as constructed by our Algorithm \ref{alg:indel-edit-apx}, where $k=1$. Matching between characters is indicated by solid red arrows, while substitutions are denoted by dotted red arrows. The deleted cells or characters are shaded in grey while the substituted cells are shaded in green.}
    \label{fig:indel-edit-apx}
\end{figure}

In order to analyze the $cost(\widetilde{\A})$, it's crucial to count the number of deletions in $X$ and $\widetilde{Y}$, along with the number of substitutions. The number of deletions in $\widetilde{Y}$ and substitutions concerning the alignment $\widetilde{\A}$ are straightforward to analyse, as discussed in the proof of Claim \ref{claim:indel-edit-apx3}. 

Our focus now shifts to counting the deletions in $X$. For that we first need to bound the following quantity. For each $1 \leq i \leq l$, let $S_i$ denote the total number of characters from $d^{X}_i$ that align with characters from $m^{\widetilde{Y}}_j$ for all $j > i$. Similar to the analysis of alignment of $m^X_i$, we will examine different cases for the alignment of $d^X_i$. Additionally, we will try to bound the quantity $S_i$ for these cases.

\begin{enumerate}
    \item If $j < i$ ($m^{\widetilde{Y}}_i$ is fully available):  
    After aligning $m^{X}_i$ to $m^{\widetilde{Y}}_i$, we begin aligning characters of $d^{X}_i$ from left to right. The leftmost unaligned character in $\widetilde{Y}$ is the first character in $\widetilde{c}_i$, so we start aligning $d^{X}_i$ to $\widetilde{c}_i$. If $|d^{X}_i| > |\widetilde{c}_i|$, then some characters in $d^{X}_i$ remain unaligned. We then align these unaligned characters of $d^{X}_i$ from left to right to $d^{\widetilde{Y}}_i$, where there can be both matchings and substitutions. If $|d^{X}_i| > |\widetilde{c}_i| + |d^{\widetilde{Y}}_i|$, we move to the next block in $\widetilde{Y}$ to align the still unaligned characters of $d^{X}_i$. In doing so, we may align characters of $d^{X}_i$ to characters in $m^{\widetilde{Y}}_j$ for some $j>i$. We define $S_i$ as the total number of characters of $d^{X}_i$ that are aligned to characters of $m^{\widetilde{Y}}_j$ for all $j>i$. In this case, $S_i$ is bounded by $\lceil \frac{|d^{X}_i| - k}{k+1}\rceil \le \lceil \frac{|d^{X}_i|}{k+1} \rceil$.

    \item if $j = i$ ($m^{\widetilde{Y}}_i$ is partially available or not available): 
    Here we have two possible cases based on the status of $m^{\widetilde{Y}}_i$:
    \begin{enumerate}
        \item if $m^{\widetilde{Y}}_i$ is partially available: In this case the leftmost unaligned character in $\widetilde{Y}$ will be the first character of $\widetilde{c}_i$ as in the previous case. Therefore we get, $S_i \le \lceil \frac{|d^{X}_i| - k}{k+1}\rceil \le \lceil \frac{|d^{X}_i|}{k+1} \rceil$.

        \item if $m^{\widetilde{Y}}_i$ is not available: Here, the leftmost unaligned character of $\widetilde{Y}$ can be any character in $b^{\widetilde{Y}}_i$, but in the worst possible scenario, the leftmost unaligned character of $\widetilde{Y}$ is the last character of  $b^{\widetilde{Y}}_i$, which is ``\$". We align the first character of $d^X_i$ with this ``\$". For aligning the remaining $|d^X_i|-1$ characters, we move to the next block in $\widetilde{Y}$. For that, every $k+1$ characters in $\widetilde{Y}$ have at most one character which is not ``\$". Thus, $S_i$ is bounded by $\lceil \frac{\abs{d^{X}_i}-1}{k+1}\rceil \le \lceil \frac{|d^{X}_i|}{k+1} \rceil$.
    \end{enumerate} 

    \item if $j > i$ ($m^{\widetilde{Y}}_i$ is not available): 
    We start aligning $d^{X}_i$ with the leftmost unaligned character in $\widetilde{Y}$. The characters of block $b^{\widetilde{Y}}_i$ are no longer available for alignment and the leftmost unaligned character can be non-$\$$ character. Considering that every $k+1$ characters in $\widetilde{Y}$ have at most one character which is not ``\$", $S_i$ is bounded by $\lceil \frac{|d^{X}_i|}{k+1} \rceil$.
    
\end{enumerate}

From the above analysis of the alignment of $d^X_i$, we have the following claim.

\begin{claim}\label{claim:indel-edit-apx}
    For each $i, S_i \le \lceil \frac{\abs{d^{X}_i}}{k+1} \rceil$.
\end{claim}

Using the above claim and the observation that the total number of deletions in $X$ is exactly $\sum_{i=1}^{l}S_i$, we can now prove the following claim.
\begin{claim}\label{claim:indel-edit-apx2}
For the alignment $\widetilde{\A}$, the total number of deletions in $X = \sum_i{S_i} < \frac{1}{k}(\LCS(X,Y)) $.
\end{claim}
\begin{proof}

    Given that the number of deletions in $X$ is $\frac{\LCS(X,Y)}{2}$ (as per Fact \ref{fact:indel-edit1}), by Claim \ref{claim:indel-edit-apx}, we have $\sum_i S_i \le \sum_i \lceil \frac{|d^{X}_i|}{k+1} \rceil < \sum_i \frac{2|d^{X}_i|}{k} \le \frac{\LCS(X,Y)}{k} = \frac{\LCS}{4} \times \varepsilon$, given $k = \frac{4}{\varepsilon}$ for any $0 < \varepsilon \le 1$.
\end{proof}

We can finally bound the cost of alignment $\widetilde{\A}$ in the following claim.

\begin{claim}\label{claim:indel-edit-apx3}
    $\ED(X,\widetilde{Y}) \le cost(\widetilde{\A}) < \widetilde{N} - n + \frac{\LCS}{2} + \frac{\LCS}{2} \varepsilon$.
\end{claim}
\begin{proof}

    Since $\widetilde{\A}$ is an $\ED$ alignment of $X$ and $\widetilde{Y}$, therefore $cost(\widetilde{\A})$ is less than the cost of an optimal $\ED$ alignment of $X$ and $\widetilde{Y}$, which is $\ED(X,\widetilde{Y})$.
    
   Using Fact \ref{fact:indel-edit2}, we know that $\ED(X,\widetilde{Y}) = $ \#deletions in $X$ + \#deletions in $\widetilde{Y}$ + \#substitutions $= (\widetilde{N}-n) + 2\times$\#deletions in $X$ + \#substitutions.

    Since, all the substitutions are in $d^{X}_i$ for all $0\le i \le l$, therefore, \#substitutions $\le \sum_i{\abs{d^{X}_i}} = \frac{\LCS(X,Y)}{2}$. Applying the bounds for \#deletions in $X$ from Claim \ref{claim:indel-edit-apx2}, we get $\ED(X,\widetilde{Y}) < \widetilde{N} - n + \frac{\LCS}{2} + \frac{\LCS}{2} \varepsilon$.

\end{proof}

\begin{proofof}{Theorem \ref{thm:indel-edit-apx}}
   From Fact \ref{fact:indel-edit3}, we have $\ED(X,\widetilde{Y}) \ge \widetilde{N} - n + \frac{\LCS(X,Y)}{2}$, as $\widetilde{Y}$ is obtained from $Y$ by inserting special character ``\$" which is not in $\Sigma$. From Claim \ref{claim:indel-edit-apx3}, we have $\ED(X,\widetilde{Y}) < \widetilde{N} - n + (1 + \varepsilon)\frac{\LCS(X,Y)}{2}$. 
\end{proofof}

\end{document}